\documentclass[12pt]{article}
\usepackage{newinutile-accg}
\usepackage{amsmath,amsfonts,amsthm,amssymb}
\usepackage{epsfig,graphics,color,calc,graphicx}
\usepackage{pstricks}
\usepackage{cite}
\usepackage{mathrsfs}

\newtheorem{theorem}{Theorem}

\newtheorem{lemma}{Lemma}
\newtheorem{remark}{Remark}

\newcommand{\rr}[1]{{\normalfont\textrm{#1}}}

\newcommand{\bb}[1]{{\mathbb{#1}}}

\newcommand{\Len}{L}
\newcommand{\uf}{v}
\newcommand{\uk}{u}
\newcommand{\Df}{K}
\newcommand{\Dk}{D}
\newcommand{\Lf}{b}
\newcommand{\Lk}{a}

\newlength{\pecettawidth}
\begin{document}
\title{Fick and Fokker--Planck diffusion law in inhomogeneous media}

\author{Daniele Andreucci}
\email{daniele.andreucci@sbai.uniroma1.it}
\affiliation{Dipartimento di Scienze di Base e Applicate per l'Ingegneria,
             Sapienza Universit\`a di Roma,
             via A.\ Scarpa 16, I--00161, Roma, Italy.}

\author{Emilio N.M.\ Cirillo}
\email{emilio.cirillo@uniroma1.it}
\affiliation{Dipartimento di Scienze di Base e Applicate per l'Ingegneria,
             Sapienza Universit\`a di Roma,
             via A.\ Scarpa 16, I--00161, Roma, Italy.}

\author{Matteo Colangeli}
\email{matteo.colangeli1@univaq.it}
\affiliation{Dipartimento di Ingegneria e Scienze dell'Informazione e
Matematica, Universit\`a degli studi dell'Aquila, via Vetoio,
67100 L'Aquila, Italy.}

\author{Davide Gabrielli}
\email{gabriell@univaq.it}
\affiliation{Dipartimento di Ingegneria e Scienze dell'Informazione e
Matematica, Universit\`a degli studi dell'Aquila, via Vetoio,
67100 L'Aquila, Italy.}


\begin{abstract}
We discuss diffusion of particles in a spatially inhomogeneous medium.
From the microscopic viewpoint we consider independent particles randomly
evolving on a lattice. We show that the reversibility condition has a discrete
geometric interpretation in terms of weights associated to un--oriented edges
and vertices.
We consider the hydrodynamic diffusive scaling that gives, as a macroscopic evolution equation, the Fokker--Planck
equation corresponding to the evolution of the probability distribution
of a reversible spatially inhomogeneous diffusion process. The geometric
macroscopic counterpart of reversibility is encoded into a tensor metrics
and a positive function.
The Fick's law with inhomogeneous diffusion matrix is obtained in the
case when the spatial inhomogeneity is associated exclusively with the
edge weights.
We discuss also some related properties of the systems like a
non--homogeneous Einstein relation and the possibility of uphill diffusion.
\end{abstract}





\maketitle

\section{Introduction}
\label{s:introduzione}
\par\noindent
The modelling of the
diffusion of a physical quantity encoded by a \emph{density field} $\rho(x,t)$
is usually constructed by assuming
a continuity equation
\begin{equation}
\label{int000}
\frac{\partial \rho}{\partial t}
=
-\nabla\cdot J
\end{equation}
expressed in terms of the flux vectorial field $J(x,t)$
and a relation between the flux and the density field.
The most popular choice is the \emph{Fick's law}
(see \cite{GDISP2006} for a very introductory discussion)
\begin{equation}
\label{int010}
J=-D\nabla \rho
\;\;,
\end{equation}
where the positive function $D$
is called \textit{diffusion coefficient}.
In general $D=D(\rho,x)$. When there is a dependence on $\rho$
we obtain a nonlinear equation.
For spatially homogeneous systems $D$ does not depend on $x$.

Let us for simplicity consider the cases of a diffusion coefficient that
does not depend on $\rho$.
In many experimental situations
\cite{MBCS2005,ABC2014,CKMS2016,CKMSS2016,CDMR,CC2017,CCM1997,BG2005,LBLO2001,L1984,S1993,S2008}
one should consider
a not constant diffusion coefficient $D(x)$. In this cases it is not clear
if Fick's law is the correct equation expressing the connection between
the density and the flux fields. A different choice is the
\emph{Fokker--Planck diffusion law}
(see the books \cite{G2009,K1981} for an introduction to the
Fokker--Planck equation)
\begin{equation}
\label{int020}
J=-\nabla(D\rho)
\end{equation}
which adds to the standard Fick's law a drift with velocity
$-\nabla D$, see Section~\ref{s:einst}.

In correspondence of these two different assumptions one finds
two possible equations for the diffusion problem
\begin{equation}
\label{int030}
\frac{\partial \rho}{\partial t}
=
\nabla\cdot(D\nabla \rho)
\end{equation}
and
\begin{equation}
\label{int040}
\frac{\partial \rho}{\partial t}
=
\Delta(D\rho)
\end{equation}
which will be respectively called the \emph{Fick} and
the \emph{Fokker--Planck diffusion equation}; note that they
reduce to the same equation if $D$ is constant.

These two equations can be studied in $\Lambda\times[0,T]$
with $\Lambda\subset\bb{R}^d$ and $T>0$ with
$D\in C^2(\Lambda)$ and with initial condition
$\rho(x,0)=\rho_0(x)\in C^2(\Lambda)$.
Possible boundary conditions are Dirichlet or Neumann conditions on
$\partial\Lambda$. In case $\Lambda$ is a parallelepiped,
it is possible to consider periodic boundary conditions.

In the applied science literature there are many situations in
which the two different points of view are assumed.
We just mention the paper \cite{SD2005} where the Fick's law is used
to study the transport of nutrients in cartilaginous
tissues and the paper \cite{MBCS2005} where it is discussed an
experiment in which a not uniform stationary density profile
is produced starting from a uniform distribution of particles
flowing inside a medium with not constant diffusion coefficient obtained
by adding gelatine to water.
This experimental
observation is obviously in contrast with the Fick's law prediction.

The fact is that, as clearly explained in \cite{S1993,SBBP1990},
the question ``what is the right generalization
of the Fick's law to inhomogeneous systems" is too naive.
A more detailed knowledge of the microscopic system is necessary
to model correctly the macroscopic behavior.
In \cite{S1993} the authors, in particular, discuss a
convincing and simple example based on two systems
in which a closed box contains a very dilute gas moving through
a dense mesh of iron wool. Model one: the iron wool
density is uniform and the box experiences a fixed temperature gradient
so that the typical particle speed varies continuously throughout
the box. Model two: the temperature is uniform, but the
iron wool density varies continuously in the box.
The systems are designed so that the effective diffusion
coefficient, which can be defined as the ratio between the square of
the mean free path and the mean free time, is the same function of
the space coordinates in the two systems.
The authors remark that, since the temperature is uniform
in box two and not uniform in box one they expect a stationary
uniform particle density distribution in box two and not uniform in
box one; indeed, they also deduce Fokker--Planck behavior
for the first model and Fick for the second.

Our work is very much in the spirit of \cite{S1993,SBBP1990}, indeed,
we assume the microscopic point of view
and prove that
two different models behave
in the hydrodynamic limit
\cite{DmP1991,KL1999}
respectively according to the Fick and the Fokker--Planck diffusion law.
In our modelling particles move in a discrete space and jump
from one site to another following an edge. We find the Fick's
behavior if the inhomogeneity is associated with edges and the Fokker--Planck
one if inhomogeneity is associated with sites.

Our modelling provides a deep physical interpretation
of the phenomenon, indeed, it suggests that
the Fokker--Planck's law is associated with locally
isotropic inhomogeneities,
whereas inhomogeneity accompanied to anisotropy results
into Fick's behavior.
More precisely,
suppose that in a small interval of time
the number of particles leaving a site of
the system is equally distributed among the edges intersecting that site,
then the macroscopic behavior is Fokker--Planck.
On the contrary,
suppose that the number of particles leaving a site are not
equally distributed among the edges intersecting that site,
but assume also that if two sites connected by an edge are occupied
by the same number of particles then the number of particles
moving along the bond in the two directions is equal.
In such a case the macroscopic behavior is Fick.
The second assumption assures that there is no preferred direction
along an edge, in particular it rules out the
possibility to have external fields acting on the system.

We note, finally, that our results are coherent with the simple example
discussed in \cite{S1993}.
Consider a small portion of volume in the box one, the number
of particles exiting the volume depends on its location
due to velocity gradient. But, since the wool mesh is
uniformly distributed, particles move with the same speed in all directions,
so that the system is locally isotropic and this, accordingly to our results,
implies the Fokker--Planck behavior.
On the other hand, in box two the non--uniformity of the
iron wool distribution breaks the local isotropy and this is
why the Fick's behavior is found.

As we mentioned above the main goal of the paper is the
derivation of the Fick and Fokker--Planck diffusion laws starting
from a microscopic model in which the spatial inhomogeneity
is differently implemented.
The paper contains also a final section in which we discuss
some relevant phenomena connected with inhomogeneous diffusion.
In particular, we note that coupling a Fick channel with a Fokker--Planck one
with suitable boundary conditions gives rise to the
phenomenon of uphill currents, in the sense that the current will
flow in the standard downhill direction in the Fick channel, namely,
from the higher density end to the lower density one,
whereas it will flow uphill in the Fokker--Planck channel.
Moreover, in the same section we discuss the validity of an
inhomogeneous Einstein relation.

The paper is organized as follows.
In Section~\ref{s:modello}
we introduce the microscopic model and discuss some elementary
properties connected to invariant measures.
In Section~\ref{s:idro} we first introduce the basic notions which
are needed to state our main result on the scaling limit which is,
indeed, stated in Section~\ref{s:idrodinamico}
and proven in Sections~\ref{s:mart}--\ref{s:uni}.
Some heuristics and numerical simulations are
given in Section~\ref{s:numerici}.
Finally, in Section~\ref{s:additional} we report some
additional remarks as the above mentioned uphill current
and Einstein relation.

\section{Models}
\label{s:modello}
\par\noindent
We discuss here the microscopic structure of our inhomogeneous media.

\subsection{Preliminaries}
\label{s:preliminari}
\par\noindent
At microscopic level we have a graph
with vertices $V$, and directed edges $E$. The corresponding set of unordered edges
is denoted by $F$. A generic directed edge is denoted by $(x,y)\in E$ while an undirected one by $\{x,y\}\in F$. We consider always finite graphs such that
if $\{x,y\}\in F$ then both $(x,y)$ and $(y,x)$ belong to $E$.

Two vertices $x,y\in V$ are said to be \emph{neighbors} if and only if
$\{x,y\}\in F$. We assume that the graph is \emph{connected}, namely,
for any pair of vertices $x,y\in V$ there exists a sequence
of unordered edges $e_1,\dots,e_n\in F$ such that $x\in e_1$, $y\in e_n$, and
$e_m\cap e_{m+1}\neq\emptyset$ for $m=1,\dots, n-1$.
For any $x\in V$ we let $C(x)\subset V$ be the set of vertices that are neighbors of $x$.
The directed graph $(V,E)$ is called strongly connected if for any pair of vertices
$x,y\in V$ there exists a directed path going from $x$ to $y$. We assume that our graphs
are always strongly connected.

\subsection{Random walks and particle systems}
\label{s:random}
\par\noindent
We consider one particle performing a Random Walk on the graph $(V, E)$ with rates
$r(x,y)>0$ when $(x,y)\in E$.
We say that the random walk is \emph{reversible} if and only if
there exists a probability measure
$\mu(x)$ on $V$ such that the
\emph{detailed balance condition}
\begin{equation}\label{db1}
\mu(x)r(x,y)=\mu(y)r(y,x)\,, \qquad \{x,y\}\in F
\end{equation}
is satisfied.
This condition can be satisfied only if $\{x,y\}\in F$ implies that
both $(x,y)$ and $(y,x)$ belong to $E$. We stress again that this will be always true. If the condition \eqref{db1} is satisfied then $\mu$ is invariant for the dynamics. This means that if the walker is distributed initially like $\mu$ its distribution does not change with time.

The \emph{inhomogeneous random walk} (IRW) is the Markov jump
process on the graph with transition rate from $x$ to $y$ given by
\begin{equation}\label{rr1}
r(x,y):=\alpha(x)Q(\{x,y\})\,.
\end{equation}
where
$\alpha:V\to\mathbb{R}_+$ and
$Q:F\to\mathbb{R}_+$ are arbitrary functions. We stress that $Q$ is a function on un-ordered edges
so that $Q(\{x,y\})=Q(\{y,x\})$.
To avoid irreducibility problems we assume that such functions are strictly positive.
Sometimes we shall consider two particular cases in which the
inhomogeneity is associated exclusively either with sites or
bonds.
The \emph{site inhomogeneous random walk} (SIRW) is the IRW with
$Q(e)=1$ for any $e\in F$ and
the \emph{edge inhomogeneous random walk} (EIRW) is the IRW with
$\alpha(x)=1$ for any $x\in V$.

\smallskip

We can pass from the case of one single particle to that of $M$
independent and indistinguishable particles letting $\eta(x)$ be
the number of particles at site $x\in V$ and considering
$\eta(x) r(x,y)$ as the rate at which one particle jumps from site
$x$ to site $y\in C(x)$.
More formally, a configuration of particles is
an element of the set $\Omega=\cup_{M=1}^{+\infty}\Omega_M$ with
$\Omega_M:=\{\eta\in\mathbb{N}^{V},\,\sum_{x\in V}\eta(x)=M\}$.
The value $\eta(x)$
is the number of particles at $x\in V$ and it is usually called the occupation variable at $x$.
If $x,y\in V$ and $\eta\in\Omega$ such that $\eta(x)\ge1$,
we denote by $\eta^{x,y}$ the configuration obtained by
$\eta$ letting one particle jump from $x$ to $y$.
This means that,
$\eta^{x,y}(x)=\eta(x)-1$ and $\eta^{x,y}(y)=\eta(y)+1$ while all the remaining occupation variables remain the same.
The stochastic evolution is encoded by the generator
\begin{equation}\label{pisiq}
\mathcal{L} f(\eta)=
\sum_{(x,y)\in E}c_{x,y}(\eta)\left[f(\eta^{x,y})-f(\eta)\right]\,,
\end{equation}
with
\begin{equation}\label{lmm}
c_{x,y}(\eta)=\eta(x)\alpha(x)Q(\{x,y\})
\end{equation}
and $f:\Omega\to\mathbb{R}$.
The trajectories $(\eta_s)_{s\in[0,t]}$ of this Markov process belong to the space $D([0,t], \Omega)$. This is the space of the maps
$\eta_{\cdot}:[0,t]\to \Omega$ that are right continuous and have limit from the left.
We endow this space by the Skorokhod topology \cite{Bil}.

In the following we will denote by $\mathbb P_\nu$ the probability measure on $D([0,t], \Omega)$ determined by the Markovian stochastic evolution
given by \eqref{pisiq} when the particles are distributed at time $0$ according to the measure $\nu$. The corresponding expected value will be denoted
by $\mathbb E_\nu$. The probability and the expected value with respect to a probability measure $\nu$ on $\Omega$ will be instead denoted
respectively by $E_\nu$ and $P_\nu$ (or simply $\nu$).

\subsection{Invariant measures}
\label{s:invariante}
\par\noindent
Let us first discuss the case of one single particle.
We claim that the class of all the reversible random walks on the graph $G$
indeed
coincides with the class of IRW.

\begin{lemma} \label{carrevw}
A random walk on $(V,E)$ is reversible if and only
if the rates of transition are of the form \eqref{rr1}.
Moreover the invariant measure is $\mu(x)=1/(\alpha(x)Z)$
where $Z=\sum_{y\in V}\alpha^{-1}(y)$ is a normalization constant.
\end{lemma}
\begin{proof}
Consider first a random walk with rates \eqref{rr1} and consider the
probability measure $\mu(x)=1/(\alpha(x)Z)$. Then the detailed balance condition \eqref{db1} holds and the random walk is then reversible and the invariant measure is $\mu$.
Conversely consider a random walk for which \eqref{db1} holds. Define then $Q(\{x,y\}):=\mu(x)r(x,y)=\mu(y)r(y,x)$ and
$\alpha(x)=\mu^{-1}(x)$. Then with this choice of the weights formula \eqref{rr1} holds and we have therefore an IRW.
\end{proof}

For the many particle system,
the dynamic conserves the total number of particles and consequently if
there are not sources
there will be a family of invariant measures
depending on the number of particles. On each subset $\Omega_M$ the
dynamics is irreducible and there will be a corresponding unique
invariant measure. This is the canonical invariant measure with $M$ particles $\nu^M$
defined by $\nu^M(\eta)=0$ if $\eta\not\in \Omega_M $ and otherwise
\begin{equation}
\label{staz000}
\nu^M(\eta)
=
\frac{1}{Z_M}
\prod_{x\in V}\frac{(\alpha(x)^{-1})^{\eta(x)}}{\eta(x)!}\,,\qquad \eta\in \Omega_M\,.
\end{equation}
By the multinomial theorem, the normalization constant is
\begin{equation}
\label{staz100}
Z_M
=
\sum_{\eta\in\Omega_M}\prod_{x\in V}\frac{\alpha(x)^{-\eta(x)}}{\eta(x)!}
=
\frac{1}{M!}\Big[\sum_{x\in V}\alpha(x)^{-1}\Big]^M
\;\;.
\end{equation}
It is easy to prove that the canonical measure is \eqref{staz000} by showing
that it satisfies the detailed balance condition
for a system of independent IRW
\begin{equation}\label{dbmulti}
\nu^M(\eta)c_{x,y}(\eta)=\nu^M(\eta^{x,y})c_{y,x}(\eta^{x,y})\,,
\end{equation}
where we recall definition \eqref{lmm} for the rates $c_{x,y}$.
We note that the average number of particles at site $x\in V$ under the
stationary measure $\nu^M$ is
\begin{equation}
\label{staz110}
E_{\nu^M}[\eta(x)]
=M\frac{\alpha(x)^{-1}}{\sum_{y\in V}\alpha(y)^{-1}}\,.
\end{equation}
Indeed we have
\begin{displaymath}
E_{\nu^M}[\eta(x)]
=
\frac{1}{Z_M}
\sum_{\eta\in\Omega_M}
 \eta(x)\prod_{y\in V}\frac{\alpha(y)^{-\eta(y)}}{\eta(y)!}
=
\frac{1}{Z_M}
\sum_{k=1}^M
 \frac{\alpha(x)^{-k}}{(k-1)!}
\sum_{\eta\in\Omega^x_{M-k}}
 \prod_{y\in V\setminus\{x\}}\frac{\alpha(y)^{-\eta(y)}}{\eta(y)!}
\end{displaymath}
where
$\Omega^x_{M-k}$ denotes the set
$\{\eta\in\mathbb{N}^{V\setminus\{x\}},\,
   \sum_{y\in V\setminus\{x\}}\eta(y)=M-k\}$.
Hence, using the expression of the partition function for $M-k$
particles on $V\setminus\{x\}$, one has
\begin{displaymath}
E_{\nu^M}[\eta(x)]
=
\frac{1}{Z_M}
\sum_{k=1}^M
 \frac{\alpha(x)^{-k}}{(k-1)!}
\frac{1}{(M-k)!}
 \Big[\sum_{y\in V\setminus\{x\}}\alpha(y)^{-1}\Big]^{M-k}
\end{displaymath}
and, making the change of variables $h=k-1$, one gets
\begin{displaymath}
E_{\nu^M}[\eta(x)]
=
\frac{1}{Z_M}
 \frac{\alpha(x)^{-1}}{(M-1)!}
\sum_{k=0}^{M-1}
 \binom{M-1}{h}
 (\alpha(x)^{-1})^h
 \Big[\sum_{y\in V\setminus\{x\}}\alpha(y)^{-1}\Big]^{M-k}
\end{displaymath}
yielding \eqref{staz110} after some straightforward algebra.

An alternative way of looking at this is by labeling the particles. Since the particles
are independent, if we distribute initially the particles independently they will be independent at any later time.
In particular considering very long times  the particles will be independent in the stationary state. Calling
$X_i\in V$ the position of the particle with label $i$ in the stationary state we have that
the variables $X_i$ are independent and each of them
has distribution
coinciding with the invariant measure of one single walker described in Lemma \ref{carrevw}. We have therefore
$$
E_{\nu^M}\left[\eta(x)\right]= E\left[\sum_{i=1}^M\delta_{X_i,x}\right]=M P(X_1=x)=M\frac{\alpha^{-1}(x)}{Z}
$$
that is exactly the right hand side of \eqref{staz110}.

\smallskip

It will be more convenient to work with the grand canonical invariant measures that are
obtained as special convex combinations of the canonical ones. The family of
grand canonical invariant measures is parameterized by a parameter related
to the averaged density.
Given a function
$\lambda(\cdot):V \rightarrow \mathbb R$ we define an associated  inhomogeneous product Poisson measure
\begin{equation}\label{psv}
\mu^{\lambda(\cdot)}(\eta)
=
\prod_{x\in V}e^{-\lambda(x)}\frac{\lambda(x)^{\eta(x)}}{\eta(x)!}\,.
\end{equation}
When $\lambda(\cdot)=\lambda$ is a constant function we call simply $\mu^\lambda$ the corresponding
homogeneous product measure.
The measure \eqref{psv} satisfies a detailed balance condition
similar to \eqref{dbmulti} provided $\lambda(x)=c\alpha^{-1}(x)$ for an
arbitrary constant $c$.
We obtain in this way a family of grand canonical invariant measures
depending on the free parameter $c$.
We note that the average number of particles at site $x\in V$ under the
measure $\mu^{\lambda(\cdot)}$ is
$E_{\mu^{\lambda(\cdot)}}[\eta(x)]=\lambda(x)$. We have therefore for the grand canonical
stationary measures $E_{\mu^{c\alpha^{-1}(\cdot)}}(\eta(x))=c\alpha^{-1}(x)$.

The canonical measures are obtained by the grand canonical ones conditioning on the total number of particles. More precisely we have
$$
\nu^M(\eta)=\mu^{c\alpha^{-1}(\cdot)}\left(\eta\Big|\sum_{x\in V}\eta(x)=M\right)\,,
$$
and the conditioning is independent from the parameter $c$ of the grand canonical measure.

\section{Scaling limits}
\label{s:idro}

\subsection{Microscopic and macroscopic observables}
\label{s:micromacro}
\par\noindent
In order to perform the scaling limits we need to introduce a general framework and some observables. We will give a microscopic and a macroscopic description of the system.
The macroscopic domain $\Lambda$ is in general a bounded domain
of $\mathbb R^d$, but to avoid dealing
with boundary conditions we consider the $d$
dimensional torus $[0,1]^d$ with periodic boundary conditions.
The discretization of the macroscopic domain
is $\Lambda_N:=(\mathbb Z/N)^d\cap \Lambda$ that will be the set of vertices denoted before as $V$, with edges between nearest neighbors sites. We call respectively $E_N$ and $F_N$  the oriented and the
un--oriented edges of the graph. We denote by $\mathcal L_N$ the generator of the process \eqref{pisiq} when the underlying graph is $(\Lambda_N,E_N)$. In general, a lower index $N$ is used to denote the fact that the graph that we are considering is the lattice $\Lambda_N$ with the corresponding edges.

A discrete vector field $\phi$ is a map $\phi:E_N\rightarrow \mathbb R$
such that $\phi(x,y)=-\phi(y,x)$. The divergence of $\phi$
is defined by
\begin{equation}\label{dd}
\nabla\cdot \phi(x)
:=\sum_{y\in C(x)}\phi(x,y)\,.
\end{equation}
A vector field $\phi$ is of gradient type if there exists a function $f:V\rightarrow \mathbb R$ such that $\phi(x,y)=f(y)-f(x)$. In this case we write $\phi=\nabla f$.

We use the same notation for the discrete and continuous gradient and divergence since they are one a discretized version of the other. To understand if the symbol means the discrete or the continuous operator we have to observe on which
object it is acting.

Given a smooth function $f:\Lambda \to \mathbb R$,
its discretized version $f_N$ on the lattice $\Lambda_N$
is defined by $f_N(x)=f(x)$, $x\in \Lambda_N$ (with abuse of notation we drop sometimes the index $N$).
Given a smooth vector field
$\psi:\Lambda\to\mathbb R^d$ a natural discretization
is obtained for example considering the line integral
\begin{equation}
\label{dvf}
\psi_N(x,y):=\int_{(x,y)}\psi(z)\cdot dl\,, \quad (x,y)\in E_N\,.
\end{equation}
We have that $\psi_N$ is a discrete vector field.

We will use repeatedly the following integration by parts formula that can be easily checked. Consider a function $f:\Lambda_N\to \mathbb R$ and a discrete vector field $\phi_N$ we have
\begin{equation}\label{byp}
\sum_{x\in \Lambda_N}f(x)\nabla\cdot \phi_N(x)=\frac{1}{2}\sum_{(x,y)\in E_N}\left(f(x)-f(y)\right)\phi_N(x,y)\,.
\end{equation}
We have also the following relationship between sums over ordered edges and unordered ones. Given two discrete
vector fields $\phi_N,\psi_N$ we have
\begin{equation}\label{EtF}
\frac{1}{2}\sum_{(x,y)\in E_N}\phi_N(x,y)\psi_N(x,y)
=
\sum_{\{x,y\}\in F_N}\phi_N(x,y)\psi_N(x,y)\,.
\end{equation}
Note that the right hand side in \eqref{EtF} is not ambiguously written since the term to be summed is symmetric in the exchange of $x$ with $y$.

Consider a collection of smooth weight functions
$Q=(Q_1,\dots ,Q_d):\Lambda\to \left(\mathbb R_+\right)^d$.
We consider a corresponding discretized version as a weight function
$Q_N$ taking values on $\mathbb R_+$ and defined on the un-oriented edges by
\begin{equation}\label{pesiq}
Q_N(\{x,y\}):=Q_i\left(\frac{x+y}{2}\right)\,, \qquad \{x,y\}\in  F_N\,,
\end{equation}
where $i$ in \eqref{pesiq} has to be fixed in such a way that $y=x\pm e^{i}$
where $e^{i}$ is the vector of modulus $N^{-1}$ and directed as the $i$ coordinate axis. Note that this discretization is very different  with respect
to \eqref{dvf} since in that case $\psi_N(x,y)$ is of order
$1/N$ while in this case $Q_N(\{x,y\})$ is of order one.

The general situation that we imagine is that the weights on the edges are the discretization $Q_N$ of positive smooth weight functions while the weights on the vertices are the discretization $\alpha_N$ of a positive smooth function.

There is a natural mathematical object to be introduced in order to describe the scaling limit of the models. This is
the empirical measure $\pi_N(\eta)$ that is a positive measure on
$\Lambda$, with finite total mass, i.e. an element of $\mathcal M^+(\Lambda)$, associated to a configuration of particles $\eta$ and defined by
\begin{equation}\label{empm}
\pi_N(\eta):=\frac{1}{N^d}\sum_{x\in \Lambda_N}\eta(x)\delta_x
\end{equation}
where $\delta_x$ is the delta measure. According to this definition, given a continuous function $f:\Lambda\to \mathbb R$ we have
$$
\int_\Lambda f\,d\pi_N(\eta)=\frac{1}{N^d}\sum_{x\in \Lambda_N}\eta(x)f(x).
$$
We endow $\mathcal M^+(\Lambda)$ with the weak topology.
We say that a sequence of configurations $\eta$ (for each $N$ we have a configuration of particles on $\Lambda_N$, for simplicity of notation the dependence on $N$ is understood) is associated to a density profile $\rho\in L^1(\Lambda)$ if
$\pi_N(\eta)\to \rho(x)dx$ where $\to$ denotes the weak convergence on $\mathcal M^+(\Lambda)$.
This means that for any continuous function $f$ (recall that $\Lambda$ is compact) we have
$$
\lim_{N\to+\infty}\int_\Lambda f\,d\pi_N(\eta)=\int_\Lambda f(x)\rho(x)dx\,.
$$

Likewise a sequence of probability measures
$\mu_N$ on the configurations of particles $\mathbb N^{\Lambda_N}$ is said to be associated with a density profile $\rho$ if for
any continuous function $f$ and for any $\epsilon >0$ we have
\begin{equation}\label{associato}
\lim_{N\to+\infty} P_{\mu_N}\left(\left|\int_\Lambda f\,d\pi_N(\eta)-\int_\Lambda f(x)\rho(x)dx\right|> \epsilon\right)=0\,.
\end{equation}

\subsection{Large deviations and free energy}
\label{ldfe}
\par\noindent
We discuss firstly the scaling limit for the empirical measure when the particles are distributed according to a grand canonical invariant measure.

We perform the computation for a generic continuous function $\lambda(\cdot)$ recalling that the grand canonical invariant measure is obtained
setting $\lambda(\cdot)=c\alpha^{-1}(\cdot)$ for a suitable $c$.
Since the measure is of product type we can discuss this problem following classic strategies and obtaining not only the scaling limit
but also the corresponding large deviations asymptotic \cite{KL1999,asLD}.  In this case it is
indeed possible to compute exactly the scaled cumulant generating function.
Let $f$ be a continuous function; we can compute
\begin{equation}
\label{cgf}
V^*(f):=\lim_{N\to +\infty}
\frac{1}{N^d}\log E_{\nu_N^{\lambda(\cdot)}}\left[e^{N^d\int_\Lambda f d\pi_N(\eta)}\right]\,.
\end{equation}
Since the invariant measure is product, \eqref{cgf} can be developed as
\begin{eqnarray}
V^*(f)
&=&\lim_{N\to +\infty}\frac{1}{N^d}
\sum_{x\in \Lambda_N}
\log \left[e^{-\lambda(x)}\sum_{k=0}^\infty\frac{\lambda(x)^ke^{f(x)k}}{k!}\right]\nonumber \\
&=& \lim_{N\to \infty}\frac{1}{N^d}\sum_{x\in\Lambda_N}\lambda(x)(e^{f(x)}-1) \nonumber \\
&=& \int_\Lambda \lambda(x)(e^{f(x))}-1) \,dx\,.
\end{eqnarray}
The last equality follows by the fact that we have in the previous step the corresponding Riemann sums.

According to general results on large deviations \cite{asLD} the corresponding
large deviations rate functional, on $\mathcal M^+(\Lambda)$ endowed with the weak convergence, is given by
\begin{equation}\label{laplat}
V(\rho)=\sup_{f\in C(\Lambda)} \left[\int_\Lambda f\,d\rho-V^*(f)\right]\,.
\end{equation}
This gives a rate functional $V$ that is $+\infty$ if
the positive measure $\rho$ is not absolutely continuous and when $\rho=\rho(x)\,dx$ we have
\begin{equation}
\label{Vro}
V(\rho)=\int_\Lambda \left[f(\rho(x))-f(\lambda(x))-f'(\lambda(x))
\left(\rho(x)-\lambda(x)\right)\right]\, dx
\end{equation}
where $f(\rho)=\rho\log\rho$ is the density of free energy for a system of independent particles. Here and hereafter with call with the same name
an absolutely continuous measure and the corresponding density.

The form of the rate functional \eqref{Vro} has a  structure similar to the one corresponding to a spatially homogeneous system. The only difference is that in \eqref{Vro} $\lambda(x)$ has to be substituted by a constant corresponding to the typical density.  Recall instead that $\lambda(x)=c\alpha^{-1}(x)$ for the inhomogeneous grand canonical measure.

The functional \eqref{Vro} plays the role of a thermodynamic potential and its probabilistic interpretation is that
roughly we have
\begin{equation}\label{ldvg}
P_{\mu_N^{\lambda(\cdot)}}\left(\pi_N(\eta)\sim \rho(x)dx\right)\simeq e^{-N^dV(\rho)}\,,
\end{equation}
where $\sim$ means closeness in the weak topology and $\simeq$ means asymptotic logarithmic equivalence (see \cite{asLD} for a precise statement).
In particular, since $V(\rho)=0$ if and only if $\rho(x)=\lambda(x)$, from \eqref{ldvg} we can deduce the scaling limit of the empirical measure
when the particles are distributed according to the invariant measure. We have indeed that $\pi_N(\eta)\to \bar\rho(x)dx=c\alpha^{-1}(x)dx$, weakly $\mu_N^{c\alpha^{-1}(\cdot)} a.e.$.

\subsection{Dynamic scaling limit}
\label{s:idrodinamico}
\par\noindent
We deduce in this section the diffusive scaling limit of many independent IRW's on the lattice $\Lambda_N$. This means that we consider a system of particles defined by the rates \eqref{lmm}. This system has a diffusive behavior and this means that we have to multiply by $N^2$ the rates of jump that corresponds to accelerate by the same scale factor the time.

Recall that we consider the situation where the weights on the lattice are inherited by discretization of $C^2$
inhomogeneities. In particular we fix some $C^2$ and strictly positive weights $Q=(Q_1,\dots ,Q_d)$ and a $C^2$ and strictly positive function
$\alpha$. The parameters of the models are fixed discretizing these functions as discussed before.

The proof of our result follows the general strategy outlined in \cite{KL1999} for gradient reversible models with the simplifying
feature that we have independent particles.  We give an outline of the proof underlying the modifications that we have to do in order
to keep into account the spatial inhomogeneity of the models.

Given $\nu_N$ and $\mu_N$ two sequences of probability measures
on the configuration of particles $\Omega$ and such that $\nu_N$ is absolutely continuous with respect to $\mu_N$ we introduce their relative entropy defined by
\begin{equation}\label{re}
H\left(\nu_N|\mu_N\right):= E_{\nu_N}\left[\log\frac{\nu_N(\eta)}{\mu_N(\eta)}\right]\,.
\end{equation}

A key mathematical object to understand the hydrodynamic behavior
of the system is the instantaneous current. This is a discrete vector field depending on configurations of particles
and representing the rate at which particles cross the bonds. If $c_{x,y}(\eta)$ is the rate at which one
particle jumps from $x$ to $y$ in the configuration $\eta$ we have that the corresponding
instantaneous current is given by
\begin{equation}\label{isc}
j_\eta(x,y):=c_{x,y}(\eta)-c_{y,x}(\eta)\,.
\end{equation}
For each fixed configuration $\eta$ this is a discrete vector field. The intuitive interpretation of the instantaneous current is the rate at which particles cross the bond $(x,y)$. Let $\mathcal N_{x,y}(t)$ be the number of particles that jumped from site $x$ to site $y$ up to time $t$ in the stochastic evolution. The current flown across the bond $(x,y)$ up to time $t$ is defined as
\begin{equation}\label{currver}
J_t(x,y):=\mathcal N_{x,y}(t)-\mathcal N_{y,x}(t)\,.
\end{equation}
This is again a discrete vector field. It is important to point out however that \eqref{currver} depends  on the whole trajectory on the time window $[0,t]$ of the system of particles while instead the instantaneous current \eqref{isc} depends just on a configuration of particles $\eta$. The importance of the instantaneous current
is based on the key observation (see for example \cite{Spohn} Section II 2.3) that
\begin{equation}\label{mart}
J_t(x,y)-\int_0^tj_{\eta(s)}(x,y)ds
\end{equation}
is a martingale.
Recalling \eqref{lmm} we have that the instantaneous current is given by
\begin{equation}\label{icr}
j_\eta(x,y)=Q(\{x,y\})\left[\alpha(x)\eta(x)-\alpha(y)\eta(y)\right]\,.
\end{equation}
Recall also that to get a non--trivial scaling limit we will accelerate the process by a factor of $N^2$ so that the instantaneous
current \eqref{icr} will be multiplied by $N^2$.

Our result is the following.
\begin{theorem}\label{ilth}
Consider a collection of IRW's associated to the discretization of $C^2$ smooth and strictly positive weights $\alpha$ and $Q$. Consider $\rho_0$ an element of $L^1(\Lambda, dx)$. Let $\nu_N$ be a sequence of probability measures on the configuration of particles $\Omega$ associated to the profile $\rho_0$ in the sense of \eqref{associato} and such that there exists a positive constant $K$ and a constant $\lambda$
such that
\begin{equation}
H\left(\nu_N|\mu^\lambda_N\right)\leq KN^d\,.
\end{equation}
When the rates in \eqref{pisiq} are multiplied by $N^2$ we have that for any $t$, for any continuous function $f$ and for any $\epsilon >0$
\begin{equation}\label{ordinario}
\lim_{N\to+\infty}\mathbb P_{\nu_N}\left(\left|\int_\Lambda f\,d\pi_N(\eta_t)-\int_\Lambda f(x)\rho(x,t)dx\right|> \epsilon\right)=0\,,
\end{equation}
where $\rho(x,t)$ is the unique weak solution of the equation
\begin{equation}\label{idro}
\left\{
\begin{array}{ll}
\partial_t\rho=\nabla\cdot\Big(\mathbb{Q}\nabla\Big(\alpha\rho\Big)\Big) \\
\rho(x,0)=\rho_0(x)
\end{array}
\right.
\end{equation}
and $\mathbb Q$ is the diagonal matrix having elements $\mathbb Q_{i,j}(x):=Q_i(x)\delta_{i,j}$.
\end{theorem}

\begin{proof}
The proof is organized into different steps.

\subsection{Preliminaries}
\label{s:mart}
\par\noindent
First of all we recall some basic facts about martingales and Markov processes (see for example \cite{KL1999} Appendix 1 Section 5).
Consider a function $g(s,\eta)$ that for each configuration $\eta$ is $C^2$ in the time variable $s$. We have that
\begin{equation}\label{primamart}
M_t:=g(t,\eta_t)-g(0,\eta_0)-\int_0^t\left(\partial_s+N^2\mathcal L_N\right) g(s,\eta_s)\,ds
\end{equation}
is a martingale.
Moreover we have that
\begin{equation}\label{secondmart}
B_t:=M_t^2-N^2\int_0^t\left[\mathcal L_N g^2(s,\eta_s)-2g(s,\eta_s)\mathcal L_N g(s,\eta_s)\right]ds
\end{equation}
is a martingale too. The $N^2$ factor is due to the rescaling of the time of the process. Since $B_0=M_0=0$ we have
mean zero martingales.

As an example consider
the discrete continuity equation for the process that is
$$
\eta_t(x)-\eta_0(x)+\nabla\cdot J_t(x)=0\,.
$$
This is true for any trajectory of the process.
Using \eqref{mart} we obtain that
\begin{equation}\label{dcei}
\eta_t(x)-\eta_0(x)+N^2\int_0^t\nabla\cdot j_{\eta(s)}(x)ds
\end{equation}
is a martingale.
A direct computation shows that
\begin{equation}\label{rih}
\mathcal L_N\eta(x)=-\nabla\cdot j_\eta(x)\,,
\end{equation}
so that \eqref{dcei} is a martingale of the form \eqref{primamart}
with $g(\eta)=\eta(x)$. We recall that in \eqref{rih} the lower index $N$ on the generator
simply stress the fact that the underlying graph is the lattice $\Lambda_N$.

Consider a smooth test function $f(s,x):\mathbb R^+\times \Lambda\to \mathbb R$ and the associated martingale
\begin{align}
M^f(t):=&\int_\Lambda f(t)\,d\pi_N(\eta_t)-\int_\Lambda f(0)\,d\pi_N(\eta_0) \nonumber\\
&-N^{-d}\sum_{x\in \Lambda_N}\int_0^tds \left(\partial_s f(s,x)\eta_s(x)+f(s,x)N^2\mathcal L_N\eta_s(x)\right)\,.\label{fff}
\end{align}
The martingale \eqref{fff} is a martingale of the form \eqref{primamart} corresponding to the function
$$
g(s,\eta)=\int_\Lambda f(s)\,d\pi_N(\eta)\,.
$$
The corresponding martingale of the form \eqref{secondmart} is given by
\begin{equation}\label{cena}
B^f(t):=\left(M^f(t)\right)^2-\int_0^t\Gamma^f(s)ds
\end{equation}
where
$$
\Gamma^f(t):=N^2\mathcal L_N\left(\int_\Lambda f(t)\,d\pi_N(\eta_t)\right)^2-2N^2\left(\int_\Lambda f(t)\,d\pi_N(\eta_t)\right)\mathcal L_N \left(\int_\Lambda f(t)\,d\pi_N(\eta_t)\right)\,.
$$
The second term (without he minus sign) on the right hand side of \eqref{cena} is called the quadratic variation of the martingale $M^f$.
A direct computation gives
\begin{equation}\label{ff}
\Gamma^f(t)= \frac{N^2}{2N^{2d}}\sum_{\{x,y\}\in F_N}Q(\{x,y\})\big(f(t,x)-f(t,y)\big)^2\left(\alpha(x)\eta_t(x)+\alpha(y)\eta_t(y)\right)\,.
\end{equation}
This is obtained by the following elementary facts and simple algebraic manipulations.
If $\{x,y\}\not \in F_N$ then
$$
\mathcal L_N\left[\eta(x)\eta(y)\right]=-\eta(x)\nabla\cdot j_\eta(y)-\eta(y)\nabla\cdot j_\eta(x)\,.
$$
We have also
$$
\mathcal L_N \left[\eta^2(x)\right]=-2\eta(x)\nabla\cdot j_\eta(x)+\sum_{y\in C(x)}\left(c_{x,y}(\eta)+c_{y,x}(\eta)\right)\,.
$$
Finally  when $\{x,y\}\in F_N$ we have
$$
\mathcal L_N\left[\eta(x)\eta(y)\right]=-\eta(x)\nabla\cdot j_\eta(y)-\eta(y)\nabla\cdot j_\eta(x)-\left(c_{x,y}(\eta)+c_{y,x}(\eta)\right)\,.
$$
Since $f,\alpha,Q$ are $C^2$, using \eqref{ff}, we have that
\begin{equation}\label{gamma}
\Gamma^f(t)\leq \frac{C}{N^{2d}}\sum_{x\in \Lambda_N}\eta_t(x)
\end{equation}
for a suitable constant $C$. This is a key estimate in our computations that is similar to the estimate that holds in the homogeneous case. This fact allows to extend the results in the homogeneous case to the non--homogeneous one.

\smallskip

With a discrete integration by parts \eqref{byp} the third  term on the right hand side of \eqref{fff} (without the minus sign) becomes
\begin{equation}\label{aida}
\int_0^tds\int_\Lambda \partial_sf(s)\,d\pi_N(\eta_s)+\frac{N^2}{2N^d}\sum_{(x,y)\in E_N}\int_0^t\left(f(s,y)-f(s,x)\right)j_{\eta_s}(x,y)\,ds\,.
\end{equation}
Using the expression \eqref{icr} of the rates and performing another discrete integration by parts, the second term in \eqref{aida} becomes
\begin{equation}\label{ms}
\frac{1}{N^d}\sum_{x\in \Lambda_N}\int_0^t\alpha(x)\eta_s(x)\left[N^2\sum_{y\in C(x)}Q(\{x,y\})\left(f(s,y)-f(s,x)\right)\right]\,.
\end{equation}
Inside squared parenthesis in the above formula we have a discrete operator acting on the test function $f$ and not depending on configurations of particles. We need to understand which is the corresponding continuous differential operator.
Since our rates are obtained by discretizing smooth functions we obtain with a Taylor expansion of $Q$ that the term
inside the squared parenthesis in \eqref{ms} can be written, up to a term $O(1/N)$, as
\begin{equation}
\label{sq}
\begin{array}{l}
{\displaystyle
\sum_{i=1}^d
\Big[Q_i(x)N^2\left(-2f(s,x)+f\left(s,x+e^i\right)+f(s,x-e^i)\right)
}
\\
\phantom{mmmmmmmmmmm}
{\displaystyle
+\frac{N}{2}
\partial_{x_i}Q_i(x)(f(s,x+e^i)-f(s,x-e^i))\Big]\,.
}
\end{array}
\end{equation}
Recall that $e^i$ is the vector associated to the $i$ Cartesian axis and having modulus $1/N$.
The expression inside the squared parenthesis in \eqref{sq} is then equal to
\begin{displaymath}
\nabla \cdot(\mathbb Q(x)\nabla f(s,x))
\end{displaymath}
up to a infinitesimal term uniform over $x$, where
the divergence ad gradient operators are the continuous ones.
We obtain, therefore, that
\begin{equation}\label{tppf}
N^{2-d}\sum_{x\in \Lambda_N}\int_0^tds\, f(s,x)\mathcal L_N \eta_s(x)=\int_0^tds\int_\Lambda \alpha\nabla\cdot\left(\mathbb Q\nabla f(s)\right)\,d\pi_N(\eta_s)+\mathcal R_N(t)\,,
\end{equation}
where the residual term $\mathcal R_N(t)$  can be bounded by
$$
|\mathcal R_N(t)|\leq \frac{Ct\int_{\Lambda}d\pi_N(\eta_0)}{N}
$$
for a suitable constant $C$. We used the fact that the dynamics is conservative and we have $\int_\Lambda d\pi_N(\eta_s)=\int_\Lambda d\pi_N(\eta_0)$ for any $s$.

Since the initial configuration is associated to an integrable profile $\rho_0$, selecting as a test function in the definition \eqref{associato} (with $\rho$ replaced by $\rho_0$ and $\mu_N$ by $\nu_N$) a function constantly equal to $1$, we deduce
\begin{equation}\label{sm}
\mathbb P_{\nu_N}\left(\sup_{0\leq s\leq t}\left|\mathcal R_N(s)\right|>\epsilon\right)\leq
P_{\nu_N}\left(\int_\Lambda d\pi_N(\eta)>\frac{\epsilon N}{Ct}\right)\stackrel{N\to+\infty}{\to}0\,, \qquad \forall \epsilon>0\,.
\end{equation}
\smallskip

The general strategy of our proof is the following.
Let us call $\mathcal P_N\in \mathcal M^1\big(D([0,t]; \mathcal M^+(\Lambda))\big)$ the probability measure corresponding to the distribution of $\left(\pi_N(\eta_s)\right)_{s\in[0,t]}\in D([0,t]; \mathcal M^+(\Lambda))$. We write shortly $\mathcal P_N=\mathbb P_{\nu_N}\cdot \pi_N^{-1}$ that means that for any measurable set $A\subseteq D([0,t]; \mathcal M^+(\Lambda))$
we have
$$
\mathcal P_N(A):=\mathbb P_{\nu_N}\left(\left(\pi_N(\eta_s)\right)_{s\in[0,t]}\in A\right)\,.
$$
We will first prove that the sequence of probability measures $\mathcal P_N$ is relatively compact. By Prohorov Theorem this
is equivalent to prove that $\mathcal P_N$ is tight.  Then we will prove that any possible limiting measure $\mathcal P^*$ of any possible converging subsequence extracted from $\mathcal P_N$
is concentrated on elements of $D([0,t]; \mathcal M^+(\Lambda))$ that are absolutely continuous for each $s\in[0,t]$ and that satisfy a suitable weak formulation of the equation \eqref{idro}. As a final step we prove uniqueness of the weak solution to \eqref{idro}. This implies that the whole
sequence $\mathcal P_N$ converges weakly to $\mathcal P^*=\delta_{\rho^{id}}$, where we call $\rho^{id}$ the unique weak solution to \eqref{idro}.
The convergence \eqref{associato} follows by the weak convergence of $\mathcal P_N$ and the fact that $\rho^{id}$ is an element of
$D([0,t]; \mathcal M^+(\Lambda))$ that is weakly continuous in the time variable.

\subsection{Tightness}
\label{s:tight}
\par\noindent
The first step consists in proving that the sequence of probability measures $\mathcal P_N$ is relatively compact. As it is discussed in \cite{KL1999} chapters 4 and 5, we need to prove relative compactness of the marginals for any fixed time and in addition we need to have a control concerning oscillations in time.

\smallskip

Since the total mass is preserved by the dynamics to prove the relative compactness of any marginal it is enough to prove it for the initial condition. Since $\Lambda$ is compact we need just to control the total mass. In particular we need to prove
\begin{equation}\label{AAh}
\lim_{A\to +\infty}\limsup_{N\to +\infty}P_{\nu_N}\left(\int_\Lambda d\pi_N(\eta)>A\right)=0\,.
\end{equation}
This is obtained by the same argument used for \eqref{sm}.

\smallskip

To control oscillations we use the Aldous criterion (see \cite{KL1999} chapter 4 Proposition 1.6). By the arguments again in \cite{KL1999} chapter 4 Section 2, we  need to prove that
\begin{equation}\label{aldous}
\lim_{\gamma\to 0 }\limsup_{N\to+\infty}\sup_{\tau}\sup_{\theta\leq \gamma}\mathbb P_{\nu_N}\left(\left|\int_{\Lambda}fd\pi_N(\eta_{\tau+\theta})-\int_{\Lambda}fd\pi_N(\eta_{\tau})\right|>\delta\right)=0\,,
\end{equation}
for any $\delta >0$ and for any $C^2$ test function $f$. In the above formula $\tau$ is varying among all the stopping times bounded by $t$ while $\theta$ is a real number varying in $[0,\gamma]$.
We use \eqref{fff} for a function $f$ that does not depend on time and we obtain that  \eqref{aldous} is true if we have
\begin{equation}\label{strnelli}
\left\{
\begin{array}{l}
{\displaystyle
\lim_{\gamma\to 0 }\limsup_{N\to+\infty}\sup_{\tau}\sup_{\theta\leq \gamma}\mathbb P_{\nu_N}\left(\left|N^{2-d}\sum_{x\in \Lambda_N}\int_{\tau}^{\tau+\theta}f(x)\mathcal L_N\eta_s(x)\,ds\right|>\delta\right)=0\,,
}\\
{\displaystyle
\lim_{\gamma\to 0 }\limsup_{N\to+\infty}\sup_{\tau}\sup_{\theta\leq \gamma}\mathbb P_{\nu_N}\left(\left|M^f(\tau+\theta)-M^f(\tau)\right|>\delta\right)=0\,.}
\end{array}
\right.
\end{equation}
The integrand in the upper condition above can be manipulated up to
the form \eqref{ms} that according to the subsequent computations can be written up to negligible terms as
$$
\frac{1}{N^d}\sum_{x\in \Lambda_N}\alpha(x)\eta_s(x) \nabla \cdot\left(\mathbb Q(x)\nabla f(x)\right)\,.
$$
By the regularity of the functions involved, the integral in the upper condition in \eqref{strnelli} is bounded by
$$
C\int_\tau^{\tau+\theta}ds\int_\Lambda d\pi_N(\eta_s)\leq C\theta \int_\Lambda d\pi_N(\eta_0)\
$$
where the inequality follows by the fact that the dynamics is conservative and $C$ is a suitable constant. Here and hereafter we denote by the same letter $C$ a generic constant that may depend just on the weight and the test functions. The values of the constants in different equations may be different. Since we have \eqref{AAh} and $\theta$ is going to zero we deduce easily the upper condition in \eqref{strnelli},
with an argument like the one for \eqref{sm}.

For the lower condition in \eqref{strnelli} we use Chebysev inequality and get
\begin{equation}\label{uuoo}
 \mathbb P_{\nu_N}\left(\left|M^f(\tau+\theta)-M^f(\tau)\right|>\delta\right) \leq \frac{\mathbb E_{\nu_N}\left(M^f(\tau+\theta)-M^f(\tau)\right)^2}{\delta^2}
\;\;.
\end{equation}
Since $\tau$ is a bounded stopping time then $M^t(\tau+\theta)-M^f(\tau)$ is again a martingale (with time parameter $\theta$)
and having quadratic variation $\int_{\tau}^{\tau+\theta}\Gamma^f(s) ds$ (see \cite{KL1999}). We have therefore that the right hand side
of \eqref{uuoo} is equal to
\begin{equation}
\frac{\mathbb E_{\nu_N}\left(\int_{\tau}^{\tau+\theta}\Gamma^f(s) ds\right)}{\delta^2}\,.
\end{equation}
Using \eqref{gamma} and the conservative property of the dynamics the last term above is bounded by
\begin{equation}
\frac{C\theta}{N^d\delta^2} E_{\nu_N}\left(\int_\Lambda d\pi_N(\eta)\right)\,.
\end{equation}
If we prove that the expected
value in the above formula is bounded then, recalling that $\theta\leq \gamma$, $\gamma\to 0$ and $N\to +\infty$, we proved also the
lower condition in \eqref{strnelli}.
This fact does not follow by the fact that $\nu_N$ is associated to an integrable profile. At this point it is relevant the entropy condition. Recall the basic entropy inequality (see for example \cite{KL1999} appendix 1 Section 8). Given two probability measures $\mu$ and $\nu$ and a function $f$ we have
\begin{equation}\label{entrineq}
E_\nu(f)\leq \beta^{-1}\left[\log  E_\mu \left(e^{\beta f}\right)+H(\nu|\mu)\right]\,,
\end{equation}
where $\beta$ is an arbitrary parameter. We apply this inequality
considering $\nu=\nu_N$, $\mu=\mu_N^\lambda$, $\beta=N^d$ and finally $f(\eta)=\int_\Lambda d\pi_N(\eta)$.
We obtain
\begin{align}
\label{basket}
 E_{\nu_N}\left(\int_\Lambda d\pi_N(\eta)\right)& \leq
\frac{1}{N^d}\left(\log  E_{\mu^\lambda_N}e^{\sum_{x\in \Lambda_N}\eta(x)}+H(\nu_N|\mu^\lambda_N)\right)\nonumber\\
& \leq e^{\lambda(e-1)}+K_0,
\end{align}
where we used the hypothesis on the relative entropy of the initial condition and the explicit form of the generating function of a Poisson distribution. We proved therefore the validity also of the lower condition in \eqref{strnelli} and we proved therefore \eqref{aldous}. The proof of tightness is concluded.

\subsection{Absolute continuity}
\label{s:abco}
\par\noindent
First of all we observe that the bound on the relative entropy for the initial distribution is still valid with respect to a slowly varying product of exponentials $\mu_N^{\lambda(\cdot)}$. This is obtained using again the entropy inequality \eqref{entrineq} with
$\nu=\nu_N$, $\mu=\mu^\lambda_N$, $\beta=1$ and $f=\log\frac{\mu^\lambda_N}{\mu_N^{\lambda(\cdot)}}$. Since we have product measures we can perform explicitly the computations obtaining
$$
\frac{1}{N^d}H(\nu_N|\mu^{\lambda(\cdot)}_N)\leq
\frac{1}{N^d}\sum_{x\in \Lambda_N}\left(\frac{\lambda^2}{\lambda(x)}+\lambda(x)-2\lambda\right)+\frac{2}{N^d}H(\nu_N|\mu^{\lambda}_N)\,.
$$
Since $\lambda(\cdot)$ is continuous and strictly positive the first term on the right hand side is a Riemann sum and converges while the second one is bounded by assumption.

Considering $\lambda(\cdot)=c\alpha^{-1}(\cdot)$ we have that
$\mu^{\lambda(\cdot)}_N$ is invariant for the dynamics and we have therefore (see \cite{KL1999} appendix 1 Section 9) that $H\left(\nu_N(t)|\mu^{\lambda(\cdot)}_N\right)$ is decreasing in time where $\nu_N(t)$ is the distribution of particles at time $t$. This means that for any $t\geq 0$ we have
$H\left(\nu_N(t)|\mu^{\lambda(\cdot)}_N\right)\leq N^d C$ for a suitable constant $C$.
This is the basic fact on which it is based the argument in \cite{KL1999} Section 1. In particular Lemma 1.6 there, should be rewritten considering in this case $I_0$ coinciding with the large deviations rate functional $V$ in \eqref{Vro}.

We deduce that any possible limit point $\mathcal P^*$ of any subsequence in $\mathcal P_N$ is concentrated on elements of
$D([0,t],\mathcal M^+)$ that are of the form $\rho(x,s)dx$ for
any $s\in [0,t]$ and $\rho(x,s)\in L^1(\Lambda)$.

\subsection{Characterization of limit points}
\label{chch}
\par\noindent
Since the sequence of probability measures $\mathcal P_N$ is relatively compact we can extract a converging subsequence. For simplicity of notation we call again $\mathcal P_N$ this converging subsequence and $\mathcal P^*$ its limit point.

Let us consider the martingale \eqref{fff}.
By the Chebysev and the Doob inequality we have
\begin{equation}\label{hulk}
\mathbb P_{\nu_N}\left(\sup_{0\leq s\leq t}|M^f(s)|> \epsilon\right)\leq\frac{4\mathbb E_{\nu_N}\left[\left(M^f(t)\right)^2\right]}{\epsilon^2}
\end{equation}
Since $B^f$ in \eqref{cena} is a martingale and $B^f(0)=0$ we have that
$\mathbb E_{\nu_N}\left[B^f(t)\right]=0$ for any $t$ and consequently
$$
\mathbb E_{\nu_N}\left[\left(M^f(t)\right)^2\right]=\int_0^t\mathbb E_{\nu_N}\left[\Gamma^f(s)\right]\,ds\,.
$$
Recalling the bounds \eqref{gamma} and \eqref{basket}
we have that the right hand side of \eqref{hulk} is bounded by
$
\frac{4Ct}{\epsilon^2 N^d}
$
for a suitable constant $C$ and this is converging to zero when $N\to +\infty$.

Let us call
\begin{equation}
\tilde M^f(t):=\int_\Lambda f(t) d\pi_N(\eta_t)-\int_\Lambda f(0) d\pi_N(\eta_0)-
\int_0^tds\int_\Lambda \left[\partial_s f(s)+ \alpha\nabla\cdot\left(\mathbb Q\nabla f(s)\right)\right]\,d\pi_N(\eta_s)\,.
\end{equation}
First we recall that by \eqref{tppf} we have
$$
M^f(t)-\tilde M^f(t)=\mathcal R_N(t)
$$
that is uniformly negligible in probability according to \eqref{sm}.

Second we observe that the map that associate to any $\pi(s)\in D([0,t], \mathcal M^+(\Lambda))$ the number
$$
\sup_{0\leq w\leq t}\left|\int_\Lambda f(w) d\pi(w)-\int_\Lambda f(0) d\pi(0)-
\int_0^w ds\int_\Lambda \left[\partial_s f(s)+ \alpha\nabla\cdot\left(\mathbb Q\nabla f(s)\right)\right]\,d\pi(s)\right|
$$
is a continuous function in the
Skorokhod topology of $D([0,t], \mathcal M^+(\Lambda))$.

Since by assumption we have that the subsequence $\mathcal P_N$ is weakly converging to $\mathcal P^*$, by Portmanteau Theorem we have for any $\epsilon >0$
\begin{equation}
\label{bns}
\begin{array}{l}
{\displaystyle
\mathcal P^*\Big(\sup_{0\leq w\leq t}\Big|\int_\Lambda f(w) d\pi(w) -\int_\Lambda f(0) d\pi(0)
-
\int_0^w\Big[\partial_s f(s)
\vphantom{\bigg\{_\}}
}\\
\phantom{mm}
{\displaystyle
+\int_\Lambda \alpha\nabla\cdot\Big(\mathbb Q\nabla f(s)\Big)\Big]\,d\pi(s)\Big|>\epsilon\Big)
 \leq \liminf_{N\to+\infty}\mathbb P_{\nu_N}\Big(\sup_{0\leq w\leq t}\Big| M^f(w)-\mathcal R_N(w)\Big|>\epsilon\Big)\,.}
\end{array}
\end{equation}
By estimates \eqref{sm} and \eqref{hulk} the right hand side in \eqref{bns} is zero and this happens for any $\epsilon >0$.
We obtain therefore that for any limiting measure $\mathcal P^*$ we have
\begin{displaymath}
\begin{array}{l}
{\displaystyle
\mathcal P^*\Big(\pi\,:\, \int_\Lambda f(w) d\pi(w)
-\int_\Lambda f(0) d\pi(0)
}
\\
\phantom{mmmmmm}
{\displaystyle
-\int_0^w\Big[\partial_s f(s)+\int_\Lambda \alpha\nabla\cdot\left(\mathbb Q\nabla f(s)\Big)\Big]\,d\pi(s)=0\,,  0\leq w\leq t\right)=1\,.
}
\end{array}
\end{displaymath}

\subsection{Uniqueness}
\label{s:uni}
\par\noindent
In the above steps we proved that any possible limit point $\mathcal P^*$ of a converging subsequence in $\mathcal P_N$ gives full measure
to elements $\pi\in D([0,t], \mathcal M^+)$ such that: $\pi(0)=\rho_0(x)dx$ (this follows by the assumption on the initial condition), for any $s\in[0,t]$ $\pi(s)\in \mathcal M^+$ is absolutely continuous $\pi(s)=\pi(s,x)dx$  and with total finite mass given by $\int_\Lambda\rho_0(x)dx$ (this follows by the conservative nature of the dynamics and the initial condition), and finally for any test function $f$ that is $C^1$ in time and $C^2$ in space we have
\begin{equation}\label{weak}
\int_\Lambda f(t)d\pi(t)-\int_\Lambda f(0)d\rho_0-\int_0^tds\int_\Lambda \left[\partial_sf(s)+\alpha\nabla\cdot\left(\mathbb Q\nabla f(s)\right)\right]d\pi(s)=0\,.
\end{equation}
Let us now show that there is a
unique $\pi(s,x)dx$ with $\pi(s)\in L^1(\Lambda)$ satisfying \eqref{weak}.
If $\pi_{1}$, $\pi_{2}$ are two solutions,
from \eqref{weak} we readily obtain for $\pi=\pi_{1}-\pi_{2}$
  \begin{equation*}
    \int_{0}^{t}
    ds
    \int_{\Lambda}
    \{
    \partial_s f(s)
    +
    \alpha
    \nabla\cdot
    (\mathbb{Q}
     \nabla
     f(s))
    \}
    d\pi(s)
    =
    0
    \,,
  \end{equation*}
  where $f$ is the solution to the Cauchy problem
  \begin{displaymath}
  \begin{array}{rl}
    \partial_{s}f
    +
    \alpha
    \nabla\cdot(
    \mathbb{Q}\nabla f)=g,
    &
    x\in\mathbb{R}^d
    \,,
    0<s<t,
    \\
    f(t,x)
    =
    0,
    &
    x\in\mathbb{R}^d.
    \end{array}
    \end{displaymath}
Here $g\in C^1([0,t]\times\mathbb{R}^d)$
is $\Lambda$--periodic, as well as all other functions,
and vanishes near $s=t$.
The existence of $f$ in the class above follows from classical
results (\cite{LSU} chapter~4 Section~5).
Then we get in fact
    \begin{equation*}
    \int_{0}^{t}
    ds
    \int_{\Lambda}
    g(s)
    d\pi(s)
    =
    0
    \,,
  \end{equation*}
  for all $g$ as above, yielding therefore $\pi=0$.

We conclude therefore that any possible limiting measure $\mathcal P^*$ needs necessarily to be $\delta_{\rho^{id}}$, Since any possible converging subsequence is converging to the same limiting measure we have that the whole sequence $\mathcal P_N$ is converging to $\delta_{\rho^{id}}$.

\smallskip

Any weak solution of the hydrodynamic equation is an element of $D([0,t], \mathcal M^+(\Lambda))$ that it is indeed weakly continuous in $t$ i.e. it is an element
of $C([0,t],\mathcal M^+(\Lambda))$. Indeed by \eqref{weak} we have for any $C^2$ function $f$
$$
\left|\int_\Lambda fd\pi(s_1)-\int_\Lambda fd\pi(s_1)\right|\leq C|s_1-s_2|\,,
$$
where the constant $C$ depends on the weights, on the function $f$ and on the total mass. The same estimate for any continuous function can be deduced by approximations.
The map that associates to any $\pi\in D([0,t], \mathcal M^+(\Lambda))$ the real number $\int_\Lambda fd\pi(s)$, for a given time $s\in[0,t]$ and a continuous function $f$,
is in general not continuous. We have however that $\mathcal P^*$ is concentrated on weakly continuous paths so that the discontinuity points of this map have $\mathcal P^*$
probability zero and by Portmanteau Theorem we deduce that $\int_\Lambda f d\pi_N(s)$ weakly converges to the constant random variable $\int_\Lambda f(x) \rho(x,s)dx$
where $\rho(x,s)$ is the solution of \eqref{idro}. Since weak convergence to a constant random variable implies convergence in probability we deduce \eqref{ordinario}.
\end{proof}

\section{Heuristics and numerics}
\label{s:numerici}
\par\noindent
In this section we discuss an heuristic argument which explains the
hydrodynamic limits stated in Section~\ref{s:modello}. Moreover,
we shall illustrate numerically the behavior of the SIRW and EIRW stochastic
models for many particles in connection with
the Fokker--Planck and Fick diffusion equations.
In this section, for notation convenience, we shall not use the
set $\Lambda_N$ as above, but we will directly work on
the graph $V=\{0,1,\dots,N\}$.

\subsection{Heuristics for the hydrodynamic limit}
\label{s:heuhyd}
\par\noindent
Consider the SIRW process on $V$ with periodic
boundary conditions for $M$ indistinguishable and independent particles.
We show that in the limit $N\to\infty$
the evolution of the Markov process density profile converges to
that of the Fokker--Planck diffusion problem provided
the \emph{diffusive scaling} is considered.
Let $a<b$ be two reals and set $z_x=a+(b-a)x/N$ so that
$z_x\in[a,b]$. Consider a positive function $D\in C^2([a,b])$ and
set $\alpha(x)=D(z_x)$
for $x\in V$.
Denote by $\eta_x(t)$ the \emph{particle profile} at time $t$,
informally speaking, $\eta_x(t)$ is the average number of particles
occupying the site $x$ at time $t$.
The change of the number
of particles at site
$x$ in a small interval $\Delta t$ can be computed as
\begin{displaymath}
\eta_x(t+\Delta t)-\eta_x(t)
=
-2\alpha(x)n_x(t)\Delta t
+\alpha(x-1)\eta_{x-1}(t)\Delta t
+\alpha(x+1)\eta_{x+1}(t)\Delta t
\;\;.
\end{displaymath}
This equality can be rewritten as
\begin{displaymath}
\frac{\eta_x(t+\Delta t)-\eta_x(t)}{\Delta t/N^2}
=
\frac{[\alpha(x+1)\eta_{x+1}(t)-\alpha(x)\eta_x(t)]
-[\alpha(x)\eta_x(t)-\alpha(x-1)\eta_{x-1}(t)]}
     {1/N^2}
\end{displaymath}
Thus, if time is rescaled as $t/N^2\rightarrow t$ (diffusive scaling),
then in the limit $N\to\infty$
the particle density profile $\eta_x(t)/(1/N)$ will tend to a function
$\rho(z,t)$ solving the equation
\begin{displaymath}
\frac{\partial \rho}{\partial t}
=
\frac{\partial^2 D\rho}{\partial z^2}
\end{displaymath}
which is the Fokker--Planck diffusion equation in $[a,b]$.

\begin{figure}
\begin{picture}(80,180)(-10,0)
\centerline{
\includegraphics[width=0.45\textwidth]{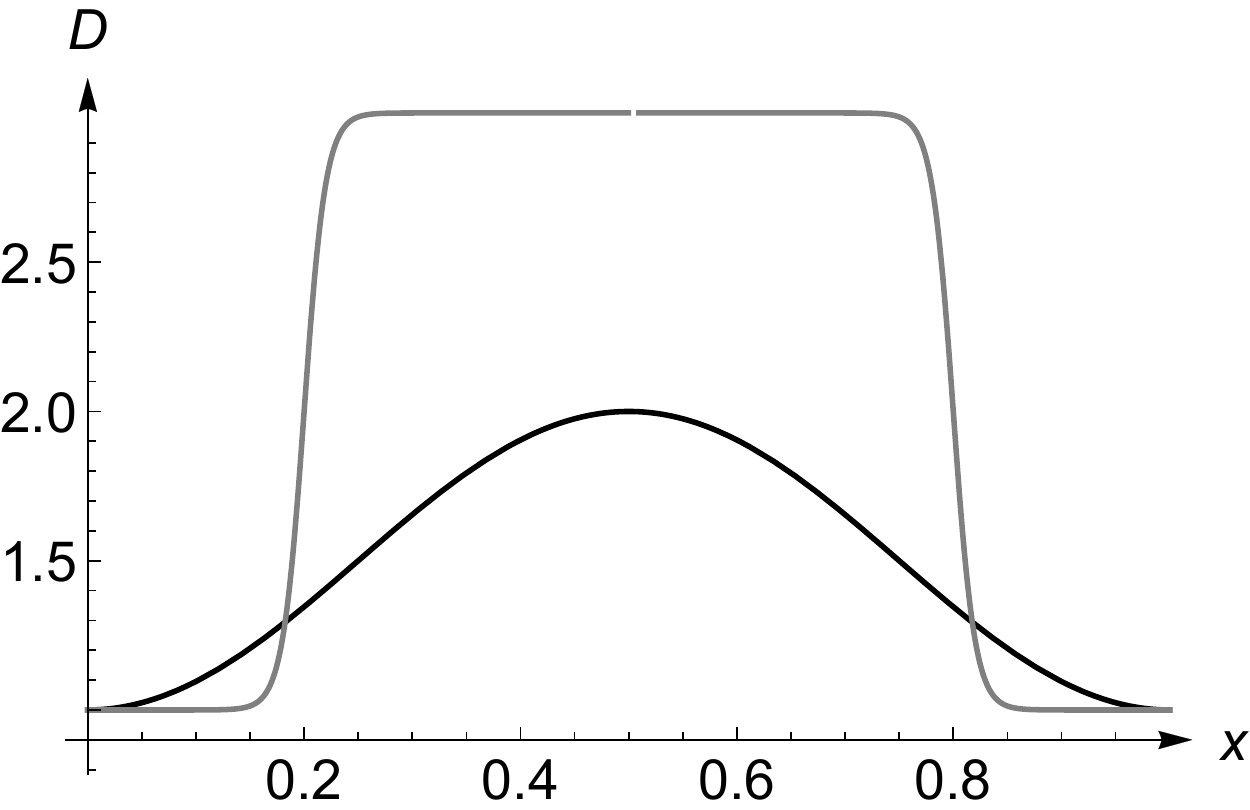}
}
\end{picture}
\caption{Diffusion coefficients \eqref{num000} (black) and
\eqref{num010} (gray).}
\label{fig:accg00}
\end{figure}

We consider the EIRW process on $V=\{0,1,\dots,N\}$ with periodic
conditions
for $M$ indistinguishable and independent particles.
and we use the same notation introduced above in the SIRW
process case.
We let
$Q(\{x,x+1\})=D((z_x+z_{x+1})/2)$
be the rate associated with the edge $\{x,x+1\}$
for $x\in V$, where $\{N,N+1\}$ is identified with
$\{N,0\}$.
The change of the number
of particles at site
$x$ in a small interval $\Delta t$ can be computed as
\begin{displaymath}
\begin{array}{l}
\eta_x(t+\Delta t)-\eta_x(t)
\\
{\displaystyle
\phantom{mm}
=
-(Q(\{x-1,x\})+Q(\{x,x+1\}))\eta_x(t)\Delta t
\vphantom{\bigg\{_\big\}}
}
\\
{\displaystyle
\phantom{mm=}
+(Q(\{x-2,x-1\})+Q(\{x-1,x\}))
 \frac{Q(\{x-1,x\})}{Q(\{x-2,x-1\})+Q(\{x-1,x\})}\eta_{x-1}(t)\Delta t
\vphantom{\bigg\{_\bigg\}}
}
\\
{\displaystyle
\phantom{mm=}
+(Q(\{x,x+1\})+Q(\{x+1,x+2\}))
 \frac{Q(\{x+1,x\})}{Q(\{x,x+1\})+Q(\{x+1,x+2\})}\eta_{x+1}(t)\Delta t
}
\\
\end{array}
\end{displaymath}
and,
hence,
\begin{displaymath}
\begin{array}{l}
\eta_x(t+\Delta t)-\eta_x(t)
\\
\phantom{mm}
=
-(Q(\{x-1,x\})+Q(\{x,x+1\}))\eta_x(t)\Delta t
+Q(\{x-1,x\})\eta_{x-1}(t)\Delta t
\\
\phantom{mm=}
+Q(\{x+1,x\})\eta_{x+1}(t)\Delta t
\;\;.
\end{array}
\end{displaymath}
This equality can be rewritten as
\begin{displaymath}
\frac{\eta_x(t+\Delta t)-\eta_x(t)}{\Delta t/N^2}
=
\frac{Q(\{x,x+1\})[\eta_{x+1}(t)-\eta_x(t)]
-Q(\{x-1,x\})[\eta_x(t)-\eta_{x-1}(t)]}
     {1/N^2}
\;\;.
\end{displaymath}
Thus, if time is rescaled as $t/N^2\rightarrow t$ (diffusive scaling),
then in the limit $N\to\infty$
the particle density profile $\eta_x(t)/(1/N)$ will tend to a function
$\rho(z,t)$ solving the equation
\begin{displaymath}
\frac{\partial \rho}{\partial t}
=
\frac{\partial}{\partial z}
\bigg(
D\frac{\partial \rho}{\partial z}
\bigg)
\end{displaymath}
which is the Fick diffusion equation.

\subsection{Numerical solution of the diffusion equations}
\label{s:num}
\par\noindent
We discuss some numerical results for the periodic
boundary condition Fick and
Fokker--Planck diffusion problem on $[0,1]\times[0,1]$
with the following choices of the diffusion coefficient:
\begin{equation}
\label{num000}
D(z)=-\frac{1}{2}\cos(2\pi z)+\frac{3}{2}
\end{equation}
and
\begin{equation}
\label{num010}
D(z)=
\left\{
\begin{array}{ll}
2+\tanh(50(z-0.2)) & z\le 0.5\\
2-\tanh(50(z-0.8)) & z> 0.5\;\;.\\
\end{array}
\right.
\end{equation}
Note that \eqref{num000} define a $C^2([0,1])$
diffusion coefficient,
whereas \eqref{num010} satisfies this condition only approximatively.

\begin{figure}
\begin{picture}(80,180)(-10,0)
\includegraphics[width=0.45\textwidth]{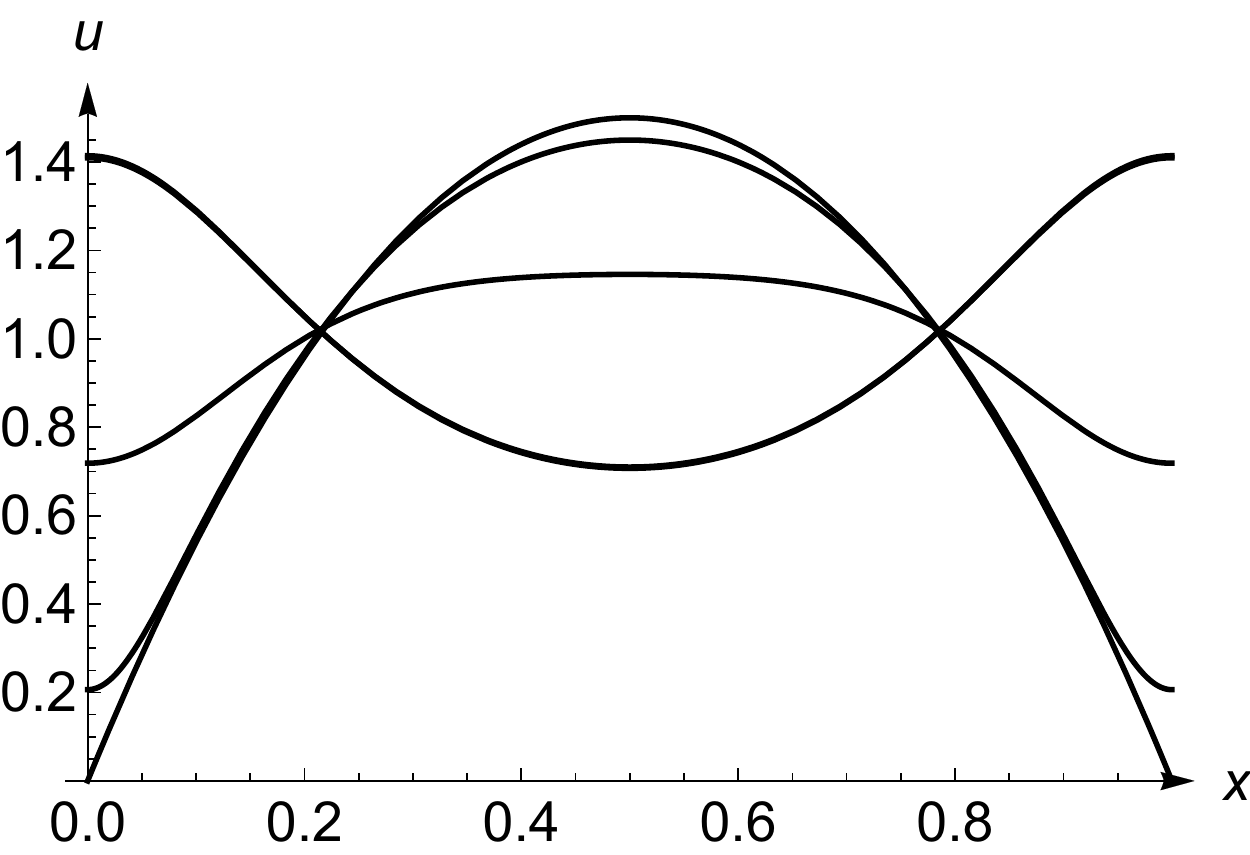}
\hskip 1. cm
\includegraphics[width=0.45\textwidth]{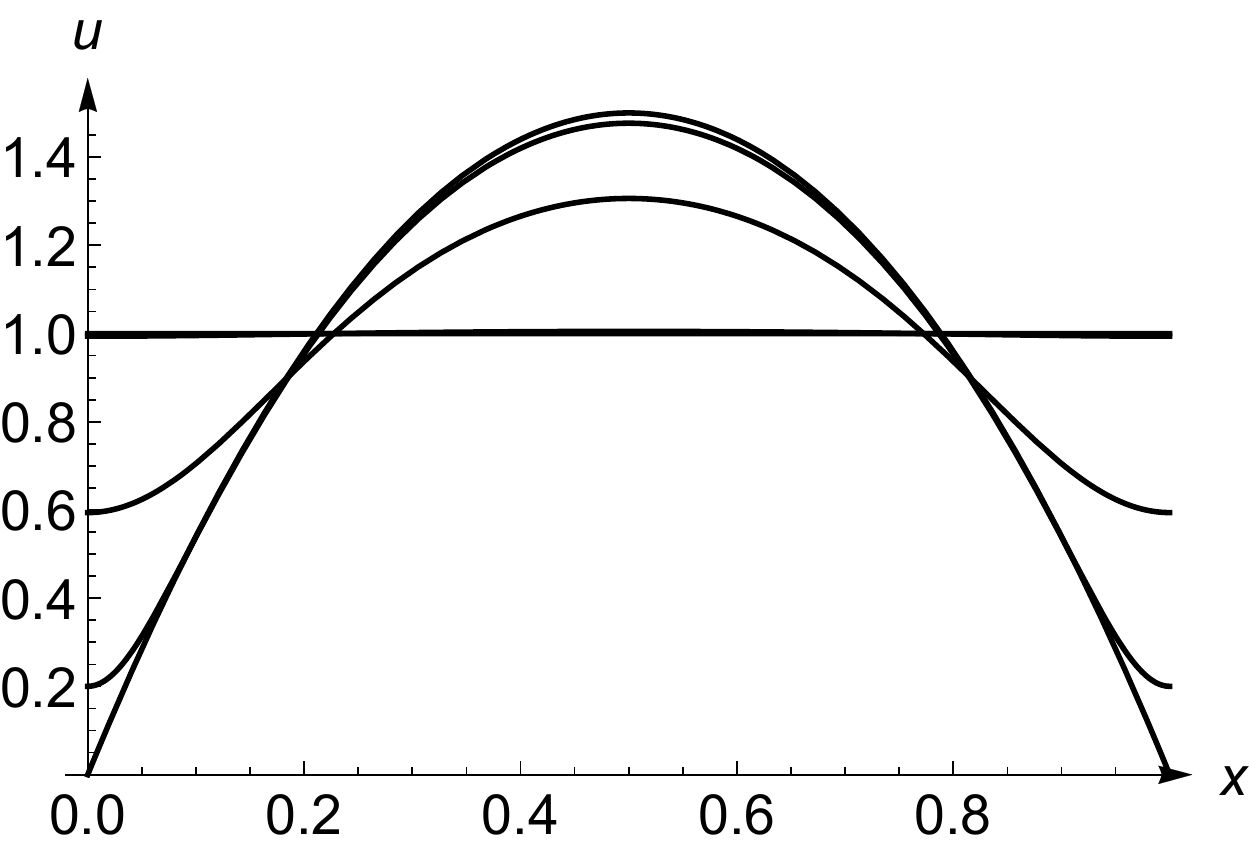}
\end{picture}
\caption{Solution of the periodic Fokker--Planck (left) and
Fick (right) problem with diffusion
coefficient \eqref{num000}.
The five curves report the solution at times
$t=0, 0.001, 0.01, 0.1, 1$,
larger the time higher the
value at the boundaries.
The two curves corresponding to times
$0.1$ and $1$ are coincident.
The initial condition is $u_0(z)=6z(1-z)$.
}
\label{fig:accg02}
\end{figure}

\begin{figure}
\begin{picture}(80,180)(-10,0)
\includegraphics[width=0.45\textwidth]{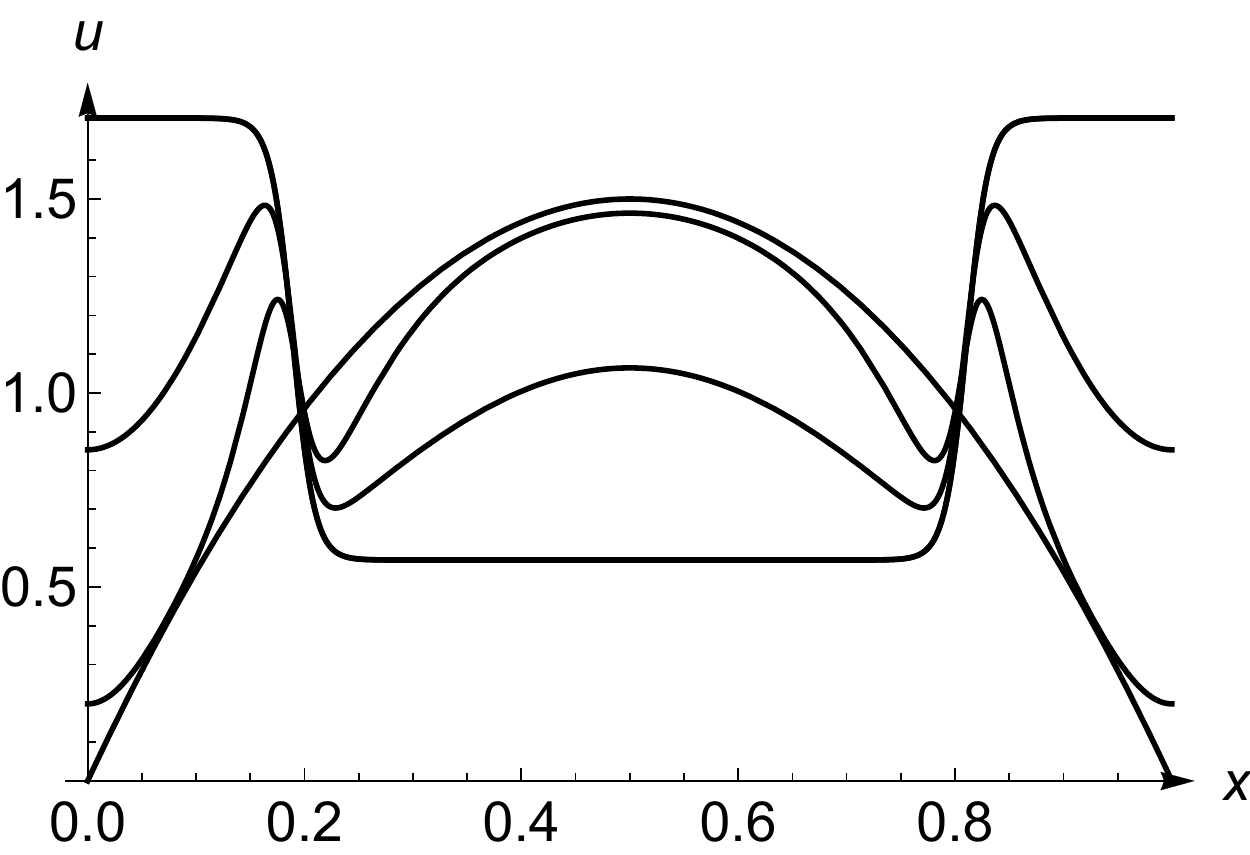}
\hskip 1. cm
\includegraphics[width=0.45\textwidth]{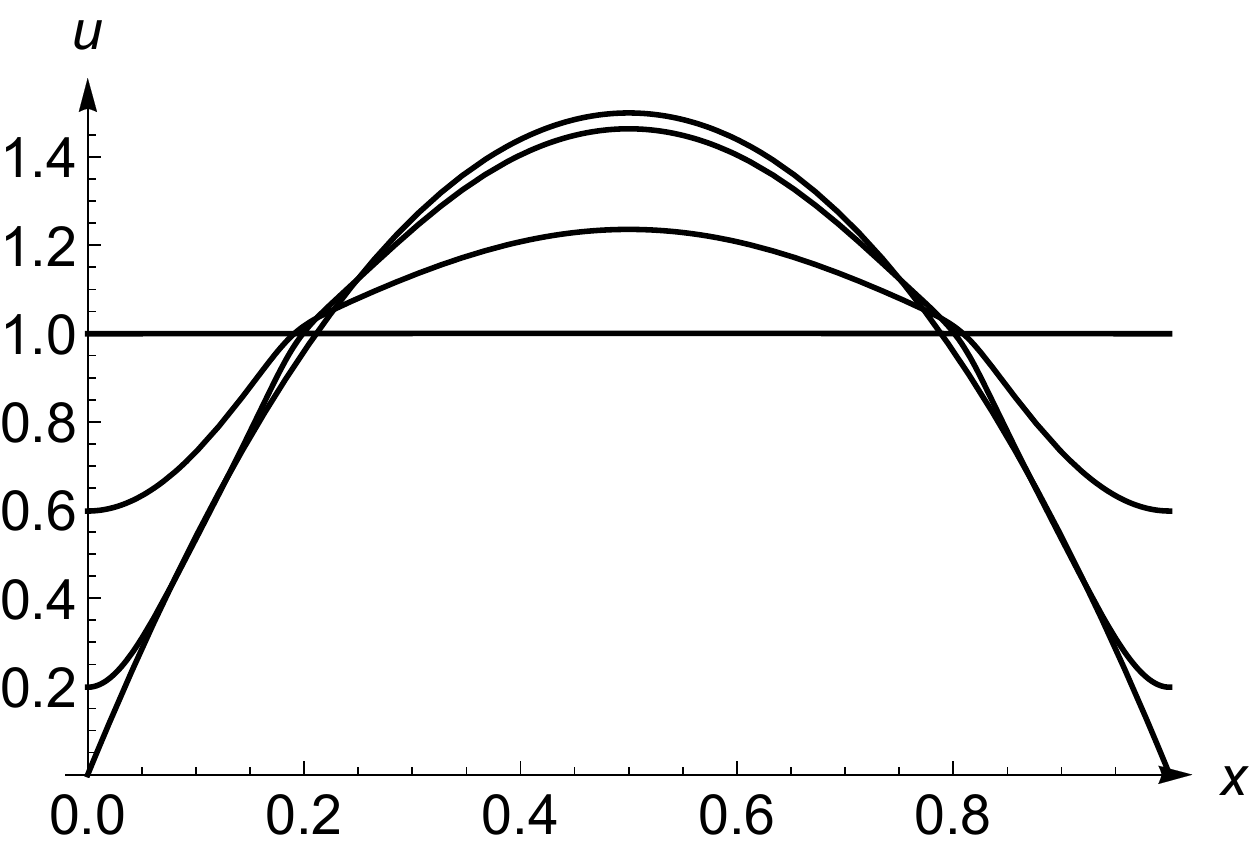}
\end{picture}
\caption{Solution of the periodic Fokker--Planck (left) and
Fick (right) problem with diffusion
coefficient \eqref{num010}.
The five curves report the solution at times
$t=0, 0.001, 0.01, 0.1, 1$,
larger the time higher the
value at the boundaries.
The two curves corresponding to times
$0.1$ and $1$ are coincident.
The initial condition is $u_0(z)=6z(1-z)$.
}
\label{fig:accg01}
\end{figure}

The numerical solution of Fick and Fokker--Planck problems
with diffusion coefficients \eqref{num000} and
\eqref{num010} are reported in
Figures~\ref{fig:accg01} and \ref{fig:accg02}.
The density field profile is reported at times
$t=0,0.001,0.01,0.1$. The profile corresponding to time
$t=0.1$ essentially coincides with the stationary solution.
The numerical solution was found using the NDSolve routine in
Mathematica.
The initial condition is $u_0(z)=6z(1-z)$ in all simulations.
We did not use a constant profile as initial condition, since that would
have been the stationary solution of the Fick diffusion process so that
no dynamics would have been observed.

Note that in the case \eqref{num010}, which mimics a
discontinuous diffusion coefficient, the Fick diffusion
problem has a constant profile as stationary solution,
whereas the Fokker--Planck problem tends to
profile
rapidly varying
in correspondence of the diffusion coefficient ``discontinuities".

The stationary solutions of the Fick and Fokker--Planck equations
can be derived explicitly.
In the Fokker--Planck case
we have that at stationarity $(Du)'$ must be constant. But, for mass
conservation, it must indeed be equal to zero, so that at stationarity
$u(z)=c/D(z)$ where the constant $c$ is such that
\begin{equation}
\label{num020}
\int_0^1\frac{c}{D(z)}\,\rr{d}z
=
\int_0^1u_0(z)\,\rr{d}z
\end{equation}
where, we recall, $u_0$ denotes the initial condition.
In the Fick case
we have that at stationarity $Du'$ must be constant. But, for mass
conservation, it must indeed be equal to zero, so that the stationarity
solution is the constant
$\int_0^1u_0(z)\,\rr{d}z$.

\subsection{SIRW process and Fokker--Planck equation}
\label{s:numsirw}
\par\noindent
We now compare the evolution of the SIRW process introduced in
Section~\ref{s:modello} to that of the Fokker--Planck diffusion equation
on $[0,1]\times[0,1]$.
The stationary profile can be discussed explicitly, indeed, in
Section~\ref{s:invariante} we have stated that at stationarity
the average number
of particles at site $x\in V$ is $b/\alpha(x)=b/D(z_x)$
with $b$ such that
\begin{equation}
\label{num040}
\sum_{x=0}^N
\frac{b}{D(z_x)}
=M
\end{equation}
where, we recall, $M$ is the total number of particles.
Comparing \eqref{num020} and \eqref{num040}
we have that, for $N$ large, $b\approx c/N$.
Hence, for $N$ large
the stationary particle density profiles
$(b/\alpha(x))/(1/N)$ of the SIRW process
is a very good approximation of the Fokker--Planck stationary solution
$c/D(z)$.

\begin{figure}
\begin{picture}(80,180)(-10,0)
\includegraphics[width=0.45\textwidth]{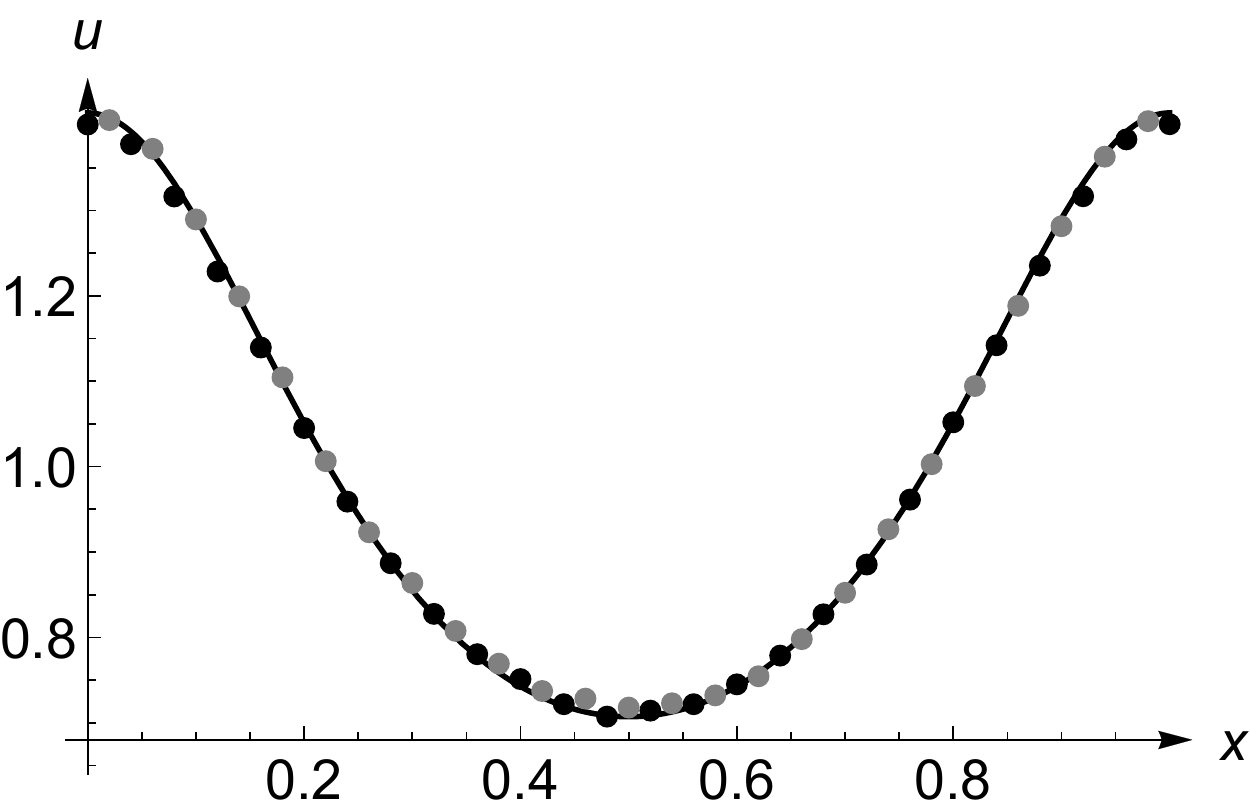}
\hskip 1. cm
\includegraphics[width=0.45\textwidth]{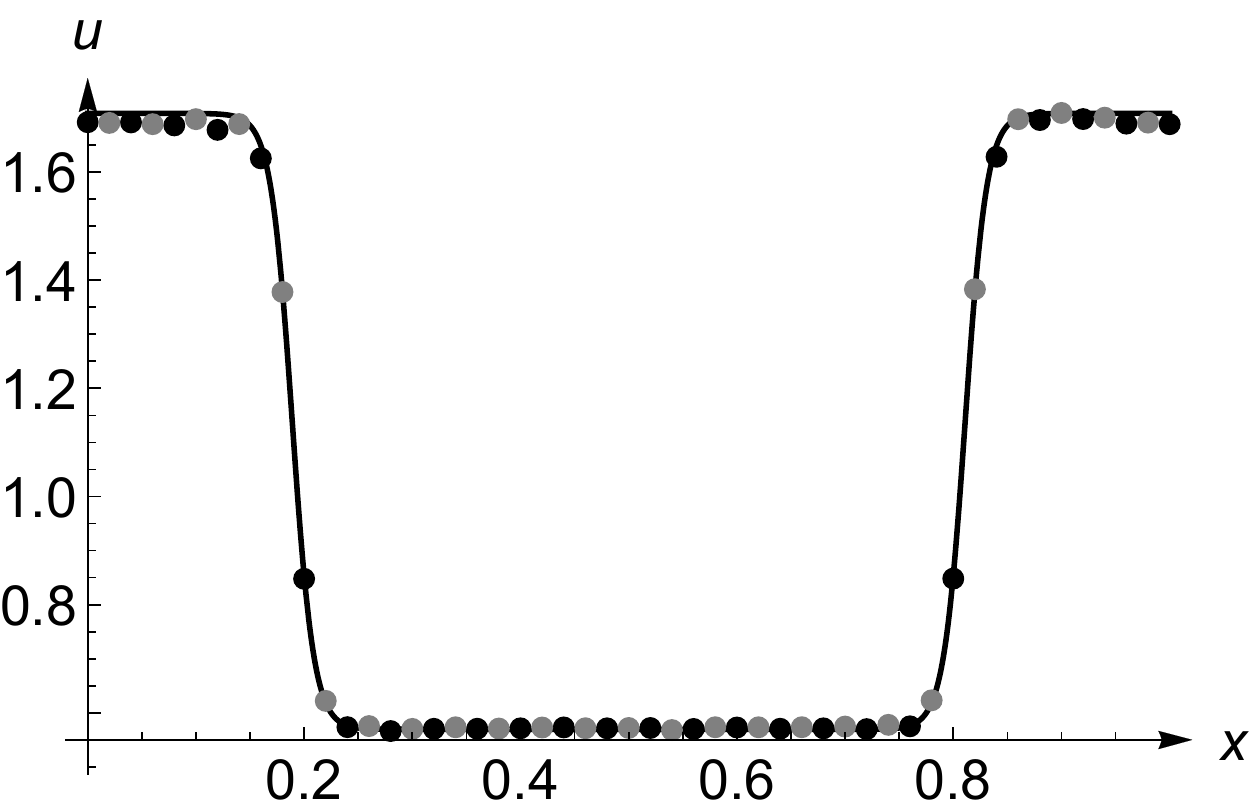}
\end{picture}
\caption{Comparison between the stationary particle profile
of the Random Walk problem multiplied times $N/M$
and the stationary solution of the
Fokker--Planck problem with diffusion
coefficient \eqref{num000} on the left and
\eqref{num010} on the right.
The black curve is the stationary solution
of the Fokker--Planck problem with initial condition $u_0(z)=6z(1-z)$, yielding
a unitary total mass.
Black and gray dots report the stationary state of the corresponding
Random Walk problem with two different initial states: a parabolic
distribution proportional to the one used for the continuous model (black)
and a uniform initial distribution (gray).
The Random Walk has been run on the lattice with $N=101$
with $M=10041$ (black) and $M=10100$ (gray).
}
\label{fig:accg03}
\end{figure}

For the time dependent results
we simulate the stochastic model as follows:
we let $z_x=x/N$ and recall $\alpha(x)=D(z_x)$ for $x\in V$.
Recalling
$\eta_x(t)$ is the number of particles
at site $x$ and time $t$,
we extract an exponential random time $\tau$ with parameter
$\sum_{x=0}^N2\alpha(x)n_x(t)$ and set the time equal to $t+\tau$.
We associate the probability
$2\alpha(y)n_y(t)/\sum_{x=0}^N2\alpha(x)n_x(t)$ to each site $y\in V$
and select at random a site according to such a distribution.
We move a particle from the selected site to one of the two
adjacent sites with probability $1/2$.

To compute the stationary particle profile we let the system evolve
for $10^3$ full sweeps (in one sweep $M$
particles are moved). Then, we average the value of the
number of particles occupying each site of the lattice by
considering one configuration each $10$ sweeps.
The numerical experiment is stopped after about $10^5$
more sweeps.

In Figure~\ref{fig:accg03} we compare the stationary
solution of the Fokker--Planck diffusion processes with
the stationary particle profile of the
Random Walk.
The stationary particle profile is divided times the spacing $1/N$ to get
the stationary particle density profile and is divided times $M$
since the Fokker--Planck diffusion equation has been solved
with an initial state having total mass equal to one.
The match is perfect.

\begin{figure}
\begin{picture}(80,180)(-10,0)
\includegraphics[width=0.45\textwidth]{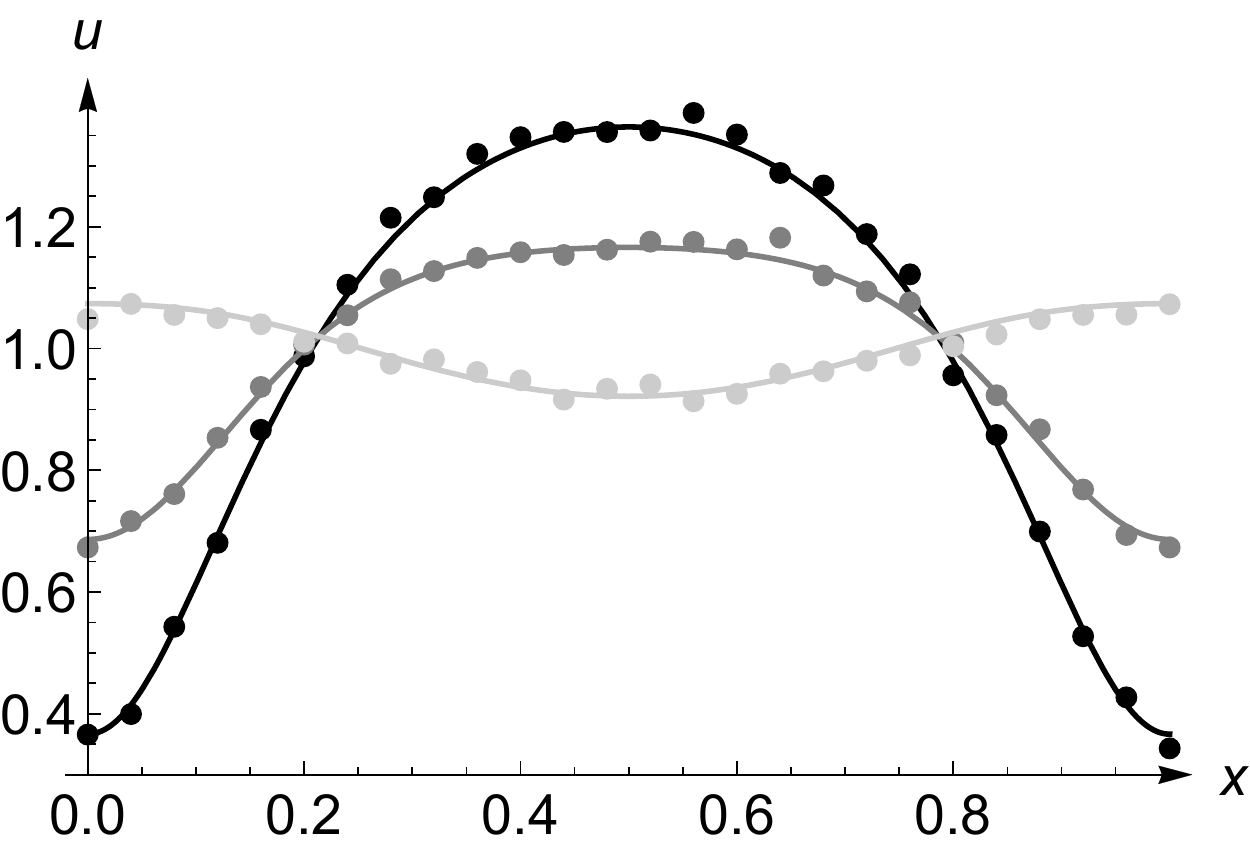}
\hskip 1. cm
\includegraphics[width=0.45\textwidth]{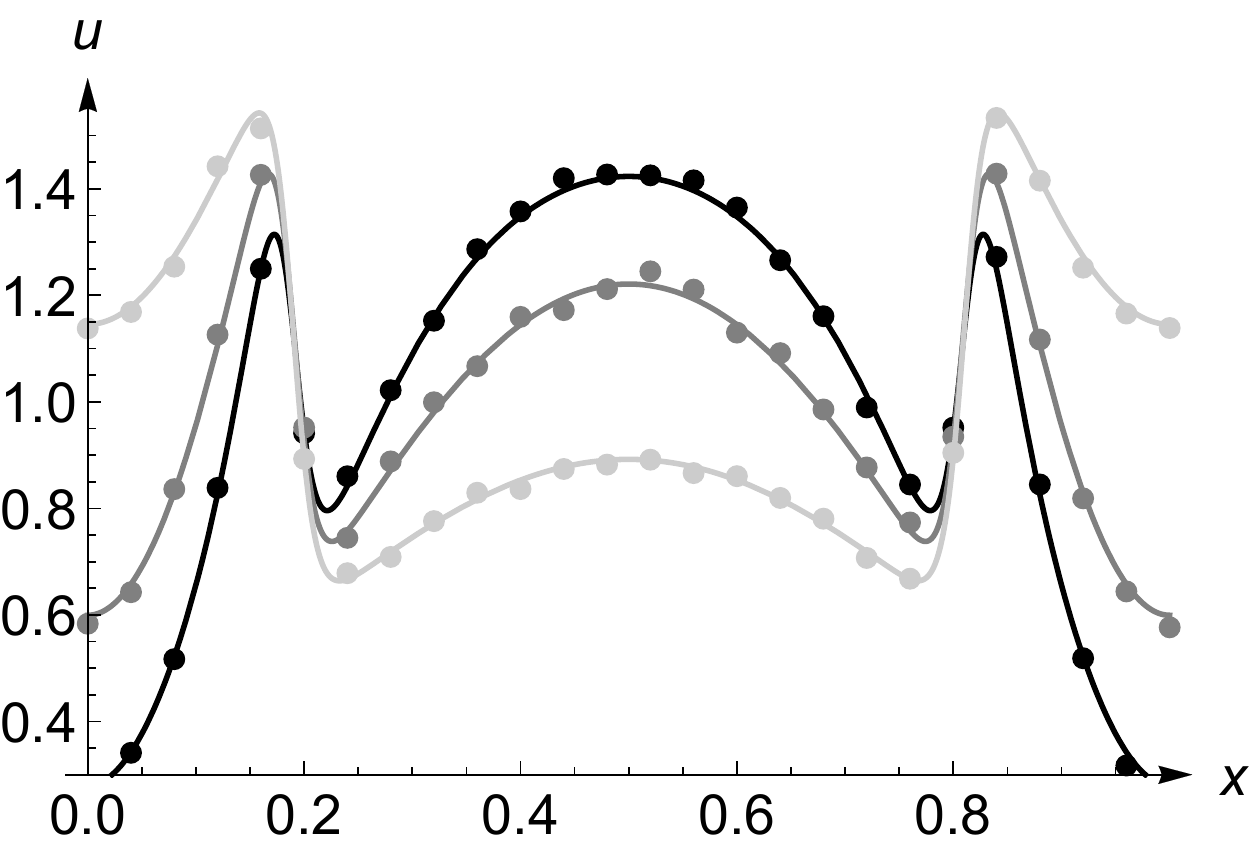}
\end{picture}
\caption{Comparison between the particle profile
of the Random Walk problem multiplied times $N/M$
and the solution of the
Fokker--Planck problem with diffusion
coefficient \eqref{num000} on the left and
\eqref{num010} on the right.
Black, gray, and light gray curves and dots refer respectively to times
$0.003005,0.009221,0.022273$ (left)
and
$0.001967,0.006207,0.015688$ (right).
Solid curves are the solution
of the Fokker--Planck problem with initial condition
$u_0(z)=6z(1-z)$, yielding a unitary total mass.
Black and gray dots report the states of the corresponding
Random Walk problem with the same initial condition.
The Random Walk has been run on the lattice with $N=101$
and $M=10041$.
}
\label{fig:accg04}
\end{figure}

In Figure~\ref{fig:accg04} we compare the evolution of the
Fokker--Planck diffusion processes with
the Random Walk particle profile.
As for the stationary state,
the Random Walk particle profile has been divided
times the spacing $1/N$ to get
the particle density profile and divided times $M$
since the Fokker--Planck diffusion equation has been solved
with an initial state having total mass equal to one.
Moreover, the time measured in the stochastic evolution
has been divided times $N^2$.
Averages have been computed by considering $50$ independent
realizations of the process and averaging the particle distribution
at equal times. The match is striking.

\subsection{EIRW process and Fick diffusion equation}
\label{s:numeirw}
\par\noindent
We now compare the evolution of the EIRW process introduced in
Section~\ref{s:modello} to that of the Fick diffusion equation
on $[0,1]\times[0,1]$. In this case the stationary
state is trivial, indeed, we compute the stationary
particle distribution profile as outlined for
the SIRW case and we find that it is constant with very
high precision.

For the time dependent results
we simulate the stochastic model as follows:
we let $z_x=x/N$ and recall
$Q(\{x,x+1\})=D((z_x+z_{x+1})/2)$ for $x\in V$,
where $\{N,N+1\}$ is identified with $\{N,0\}$.
Recalling
$\eta_x(t)$ is the number of particles
at site $x$ and time $t$,
we extract an exponential random time $\tau$ with parameter
$\sum_{x=0}^N(Q(\{x-1,x\})+Q(\{x,x+1\}))\eta_x(t)$
and set the time equal to $t+\tau$.
We associate the probability
$(Q(\{y-1,y\})+Q(\{y,y+1\}))\eta_y(t)/
\sum_{x=0}^N(Q(\{x-1,x\})+Q(\{x,x+1\}))\eta_x(t)$
to each site $y\in V$
and select at random a site according to such a distribution.
We move a particle from the selected site, say $y$,
to the left
with probability
$Q(\{y-1,y\})/(Q(\{y-1,y\})+Q(\{y,y+1\}))$
and
to the right
with probability
$Q(\{y,y+1\})/(Q(\{y-1,y\})+Q(\{y,y+1\}))$.

\begin{figure}
\begin{picture}(80,180)(-10,0)
\includegraphics[width=0.45\textwidth]{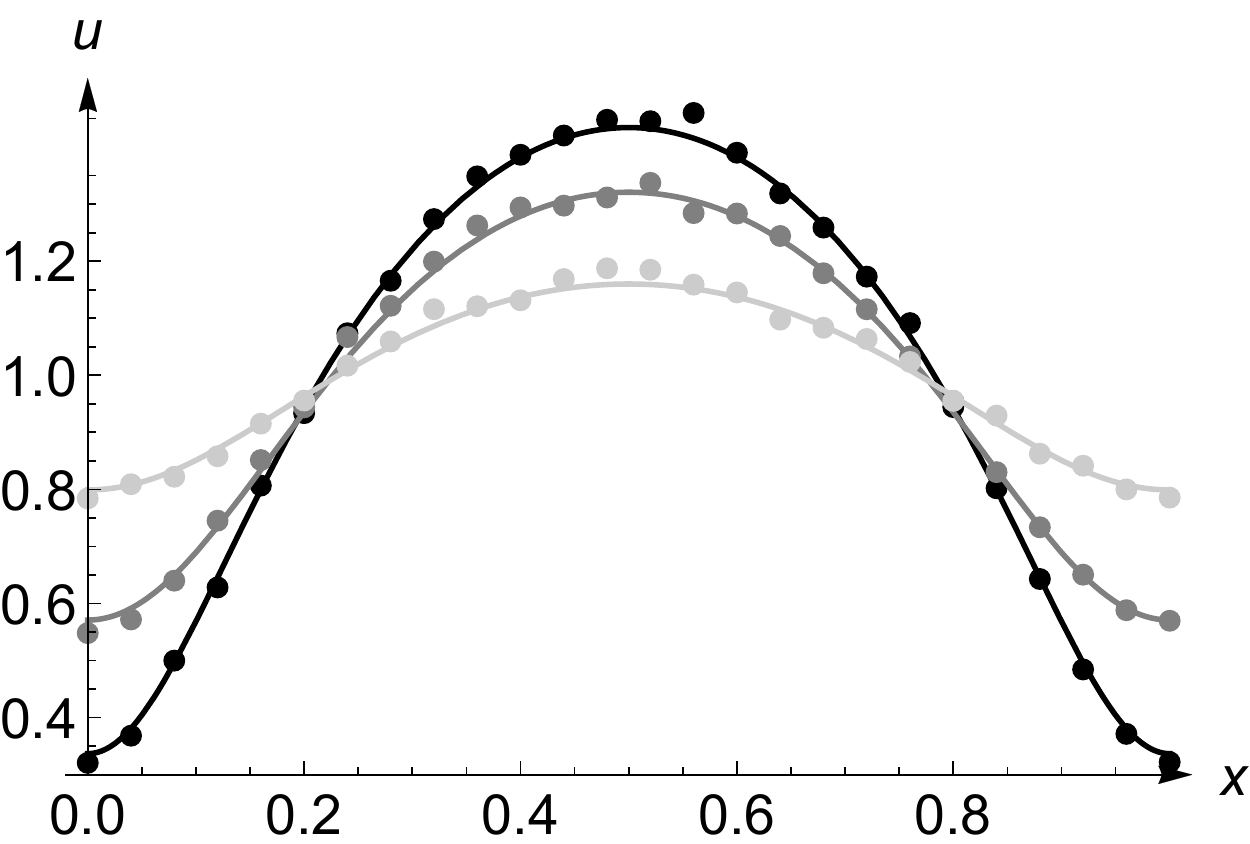}
\hskip 1. cm
\includegraphics[width=0.45\textwidth]{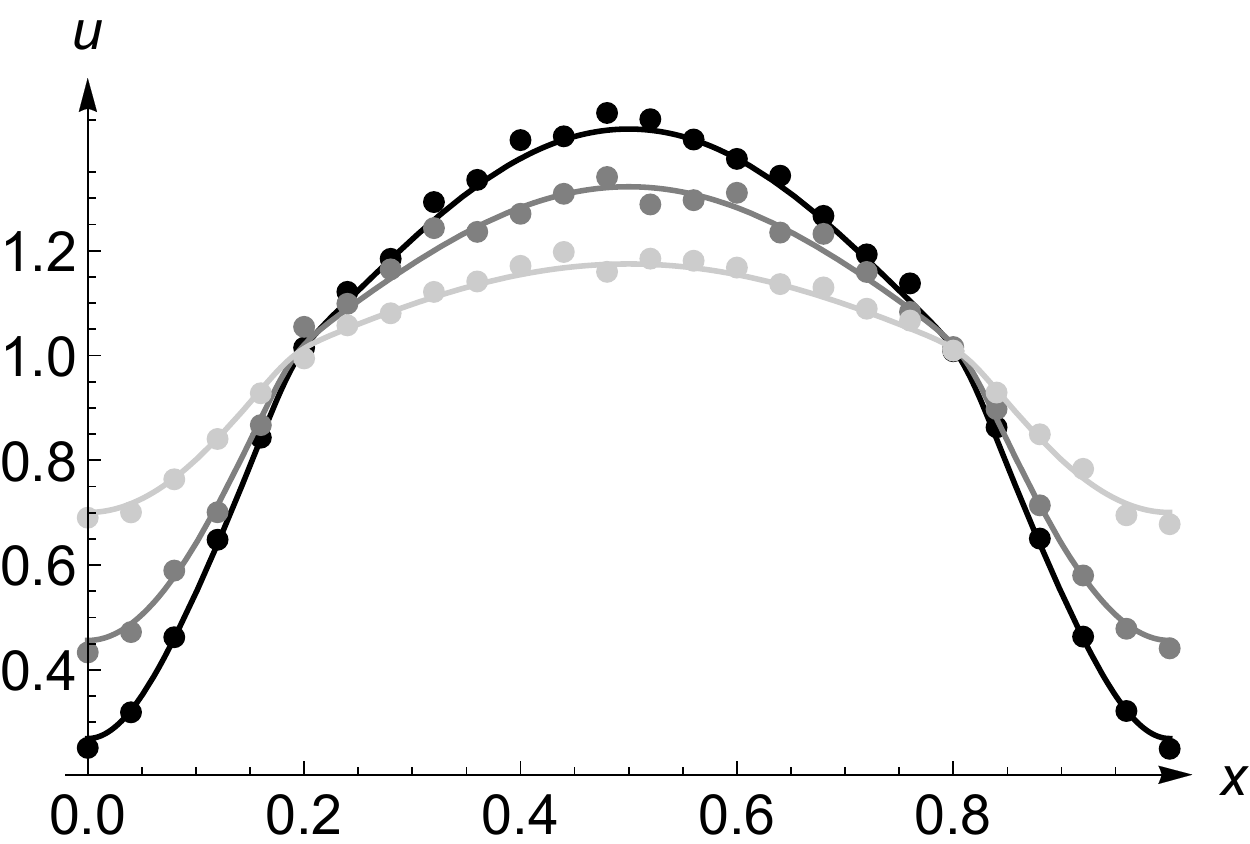}
\end{picture}
\caption{Comparison between the particle profile
of the Random Walk problem multiplied times $N/M$
and the solution of the
Fick problem with diffusion
coefficient \eqref{num000} on the left and
\eqref{num010} on the right.
Black, gray, and light gray curves and dots refer respectively to times
$0.002991,0.009102,0.021696$ (left) and
$0.001916, 0.005856, 0.014081$ (right).
Solid curves are the solution
of the Fick problem with initial condition
$u_0(z)=6z(1-z)$, yielding a unitary total mass.
Black and gray dots report the states of the corresponding
Random Walk problem with the same initial condition.
The Random Walk has been run on the lattice with $N=101$
and $M=10041$.
}
\label{fig:accg05}
\end{figure}

In Figure~\ref{fig:accg05} we compare the evolution of the
Fick diffusion processes with
the Random Walk particle profile.
As for the stationary state,
the Random Walk particle profile has been divided times the
spacing $1/N$ to get
the particle density profile and divided times $M$
since the Fick diffusion equation has been solved
with an initial state having total mass equal to one.
Moreover, the time measured in the stochastic evolution
has been divided times $N^2$.
Averages have been computed by considering $50$ independent
realizations of the process and averaging the particle distribution
at equal times. The match is striking.

\section{Miscellany}
\label{s:additional}
\par\noindent
In this section we collect some interesting remarks on the behavior
of the system that we have studied above.

\subsection{Einstein relation}
\label{s:einst}
\par\noindent
A very general modelization of the presence of an external field
is obtained perturbing the rates as follows. Let $\phi:\Lambda \to \mathbb R^d$ be a smooth vector field that acts on particles. The action of the field is encoded in the perturbed transition rates that are defined as
\begin{equation}\label{ratesfi}
c^\phi_{x,y}(\eta):=c_{x,y}(\eta)e^{\phi_N(x,y)}\,,
\end{equation}
where $\phi_N$ is the discretization \eqref{dvf} of the vector field. Rates that correspond to movements of the particles with an associate positive work of the field are enhanced while instead rates that correspond to movements of the particles with an associate negative work of the field are decreased.

\smallskip
Let us first discuss the influence of an external field in the case of spatially homogeneous models \cite{MFTReview}.
The hydrodynamic scaling limit of diffusive particle systems under the action of a weakly asymmetric external field is associated to equations of the form
\begin{equation}\label{idroM}
\partial_t\rho=\nabla\cdot \left(\mathbb D(\rho)\nabla \rho\right)-2\nabla\cdot \left(\mathbb M(\rho)\phi\right)\,.
\end{equation}
The symmetric and positive definite matrix $\mathbb D$ is the diffusion matrix while the symmetric and positive definite matrix
$\mathbb M$ is the mobility matrix. For independent particles we have that the diffusion matrix coincides with the identity matrix $\mathbb D=\mathbb I$ while instead $\mathbb M=\rho\mathbb I$.

In the homogeneous case a relevant thermodynamic relationship is the so called Einstein relation between the diffusion matrix and the mobility given by
\begin{equation}\label{Ein}
\mathbb D(\rho)=\mathbb M(\rho)f''(\rho)\,,
\end{equation}
that says that the two matrices $\mathbb D$ and $\mathbb M$ are
proportional and the proportionality factor
is the second derivative of the density of free energy $f$ (that is $f(\rho)=\rho\log\rho$ in the independent particles case as discussed after \eqref{Vro}).

Let us now move to the spatial inhomogeneous case.
An interesting way of writing the hydrodynamic equation \eqref{idro} is obtained computing the gradient appearing there, getting
$$
\partial_t\rho=\nabla\cdot\left(\alpha\mathbb Q\nabla\rho\right)
+\nabla\cdot\left(\alpha\rho\mathbb Q\nabla\log\alpha\right)\,.
$$
It is very natural to interpret this equation introducing
the space dependent diffusion matrix $\mathbb D(x,\rho)=\alpha(x)\mathbb Q(x)$ and the space dependent mobility matrix $\mathbb M(x,\rho)=\alpha(x)\rho(x)\mathbb Q(x)$. Note that they satisfy the Einstein relation for each $x\in \Lambda$. Indeed recalling that the density of free energy is $f(\rho)=\rho\log\rho$ for independent particles we have
$$
\mathbb D(x,\rho)=\mathbb M(x,\rho)f''(\rho)\,, \qquad \forall x\in\Lambda\ , \forall \rho\,.
$$
With this identification we have that the inhomogeneity determines
space dependent diffusion and mobility matrices. The form of these matrices depend both on the weights on the edges and on the weights on the vertices.  The spatial inhomogeneity of the material generates however also an external field  that depends just on the site inhomogeneity. This external field is
exactly $-(1/2)\nabla\log\alpha$.

We show that this interpretation is correct. This is done switching on a weak external field and showing that the hydrodynamic equation is modified
with the appearance of a term proportional to the mobility matrix $\mathbb M(x,\rho)$ like in the homogeneous case \eqref{idroM}.
In presence of an external field the rates are modified according to \eqref{ratesfi} and correspondingly the instantaneous current becomes
\begin{equation}\label{istfi}
j_\eta^\phi(x,y)=c_{x,y}(\eta)e^{\phi_N(x,y)}-c_{y,x}(\eta)e^{\phi_N(y,x)}\,.
\end{equation}
Recall that the values of $\phi_N$ are infinitesimal \eqref{dvf} so that we have
$$
e^{\phi_N(x,y)}=1+\phi_N(x,y)+o(1/N)\,.
$$
The instantaneous current is therefore
\begin{equation}\label{bu}
j_\eta^\phi(x,y)=j_\eta(x,y)+\left(c_{x,y}(\eta)+c_{y,x}(\eta)\right)\phi_N(x,y)+o(1/N)\,.
\end{equation}
Substituting \eqref{bu} to $j_\eta$ in the second term in \eqref{aida}
and ignoring negligible terms we obtain the extra factor
$$
\frac{1}{N^d}\sum_{x\in \Lambda_N}\int_0^tds\, \alpha(x)\eta_s(x)\Big[N^2\sum_{y\in C(x)}
Q(\{x,y\})(f(s,y)-f(s,x))\phi_N(x,y)\Big]\,.
$$
With computations similar to the ones in the proof of Theorem \ref{ilth} we have that the term inside squared parenthesis in the above formulas
coincides up to uniform infinitesimal terms with
$$
2 \mathbb Q (x)\phi(x)\cdot \nabla f(x,s)\,.
$$
This means that the hydrodynamic equation in presence of a weak external field becomes
$$
\partial_t\rho=\nabla\cdot \left(\mathbb D(x,\rho)\nabla\rho\right)
-2\nabla\cdot\left(\mathbb M(x,\rho)\left(\phi
-\frac 12\nabla\log\alpha\right)\right)\,.
$$
We deduce that $\mathbb M(x,\rho)$ plays the role of the mobility matrix and we obtain
a version of the Einstein relation in the non--homogeneous framework.

\subsection{Alternative proof}
\label{s:dimo-alt}
\par\noindent
Since we are considering a system of independent particles we can obtain an alternative  proof under some special initial conditions.
In particular we consider the case when the initial condition is obtained with identical particles distributed independently. Note that Theorem
\ref{ilth} covers much more general initial conditions.
In this special case, the collective behavior of the occupation variables can be deduced by the scaling behavior of one single particle. We could however not find a specific
reference for the scaling limit of one single IRW. The following is a sketch of the general argument that can be used once the scaling limit of one single IRW is established.

Consider the initial condition $\rho_0$ in the hydrodynamic equation \eqref{idro} and define the corresponding
probability measure $\hat\rho_0(y)=\rho_0(y)/\int_\Lambda\rho_0(x)dx$.
We consider a sequence of probability measures $p_N$ on $\Lambda_N$ such that
$$
\sum_{x\in \Lambda_N}p_N(x)\delta_x\to \hat\rho_0(y)dy\,,
$$
where the convergence is the weak one.

A simple generalization of the law of large numbers says the following. Suppose that for each natural number $N$ we have a random variable $Y^N_1$ taking values on a Polish space $\mathcal A$
and such that the law of $Y^N_1$ is converging weakly to $\gamma\in \mathcal M^1(\mathcal A)$ when $N$ diverges. We called $\mathcal M^1(\mathcal A)$ the set of probability measures on $\mathcal A$ with the Borel sigma algebra. For each $N$ let us consider $(Y^N_i)_{i\in \mathbb N}$ be a collection of i.i.d. random variables each of them having the same distribution of $Y^N_1$. Then we have that
\begin{equation}\label{ugs}
\frac{1}{N}\sum_{i=1}^N\delta_{Y^N_i}\to \gamma
\end{equation}
where the convergence is the weak one in probability (indeed even a.e.). More precisely the above statement
means that for any continuous and bounded function $f:\mathcal A\to \mathbb R$ we have
$$
\lim_{N\to+\infty}P\left(\left|\frac{\sum_{i=1}^Nf(Y^N_i)}{N}-\int_\mathcal Ad\gamma(a)f(a)\right|>\epsilon\right)=0\,, \qquad \forall\,\epsilon >0\,.
$$

We consider at time zero
$\left(\int_\Lambda \rho_0(x)dx \right)N^d$ particles independently distributed and each of them distributed on $\Lambda_N$ according to $p_N$. Let $X^N_i(0)$ be the random position in $\Lambda_N$ at time $0$ of the particle number $i$.
We consider $D([0,t]; \Lambda)$ the
Skorokhod space of trajectories. The trajectory of the particle number $i$ is denoted by $X^N_i(\cdot):=(X^N_i(s))_{s\in [0,t]}$. This is a random variable taking values on $D([0,t]; \Lambda)$. We consider each particle evolving with an IRW with rates of jump accelerated by a factor of $N^2$.

We assume in this argument that the law of the trajectory of one single particle converges to the law $\mathbb P_{\hat\rho_0}$ of a diffusion process (see next Section \ref{s:revdif} ) with  initial distribution $\hat\rho_0$ and
Kolmogorov evolution equation for the distribution given by the hydrodynamic equation \eqref{idro} (with initial condition $\hat\rho_0$). This is an assumption because we could not find a precise reference for this result.

We have therefore the convergence \eqref{ugs} that in this specific case implies that a.e., and therefore in probability, we have
$$
\gamma_N:=\frac{1}{\left(\int_\Lambda \rho_0(x)dx\right)N^d}\sum_{i=1}^{\left(\int_\Lambda \rho_0(x)dx\right)N^d}\delta_{X^N_i(\cdot)}\stackrel{N\to+\infty}{\to} \mathbb P_{\hat\rho_0}\,.
$$
Consider a continuous and bounded function $f:\Lambda\to \mathbb R$ and the functional $F:D([0,t],\Lambda)\to \mathbb R$ defined by
$F(X(\cdot)):=f(X(s))$ where $s\in [0,t]$ is a fixed time.
The functional $F$ is not continuous with respect to the Skorokhod
topology. We have however that under the probability measure $\mathbb P_{\hat\rho_0}$ the set of discontinuous points of this functional has probability zero. This is because the probability
$\mathbb P_{\hat\rho_0}$ is concentrated on continuous paths. We can therefore deduce by Portmanteau Theorem that a.e.,
and therefore in probability, we have the convergence
\begin{equation}
E_{\gamma_N}(F)\stackrel{N\to +\infty}{\to}\mathbb E_{\hat\rho_0}\left(F\right)=\int_\Lambda \hat\rho(x,s)f(x) \,dx
\end{equation}
where $\hat\rho(s)$ is the solution of \eqref{idro} with initial condition $\hat\rho_0$.
We can deduce the hydrodynamic behavior of the model observing that
$$
\int_\Lambda fd\pi_N(\eta_s)
=\frac{1}{N^d}\sum_{i=1}^{\left(\int_\Lambda \rho_0(x)dx \right)N^d}f(X^N_i(s))
=
\left(\int_\Lambda \rho_0(x)dx \right) E_{\gamma_N}(F)\,.
$$

\subsection{Reversible diffusions}
\label{s:revdif}
\par\noindent
As we observed in the previous section, in the case of independent particles the hydrodynamic equation describing
the collective behavior of several particles is linear and coincides with the equation of the evolution
of the probability distribution of one single particle. Since the scaling limit of one single particle
is a diffusion process and since our discrete models are reversible it is natural to compare the class of hydrodynamic equations that we obtained
with the possible Fokker Plank equations associated to reversible diffusions.

At the microscopic level we obtained that the reversibility condition has a geometric interpretation. We have indeed that the models are reversible if and only if
the rates are chosen according to some weights associated to the edges and the vertices of the graph (see Lemma \ref{carrevw}). In the case of continuous diffusion process we have a similar geometric characterization of reversibility, indeed reversible diffusions can be parameterized by a positive function
and a symmetric and positive definite matrix, that can be interpreted as
the metric tensor. These are the continuous counterparts of the discrete weights on the graph.

We refer to \cite{G2009, K1981} for the basic facts about diffusion processes. For simplicity we consider the processes on $\mathbb R^d$
instead that on the torus.
Consider a diffusion process of the form
\begin{equation}\label{stoceq}
dX_t= A(X_t)dt + \mathbb B(X_t)dW_t
\end{equation}
where $A=(A_1(x), \dots ,A_d(x))$ is a smooth vector field, $\mathbb B(x)$ is a $d\times d$ matrix smoothly depending on $x$ and $W=(W_1, \dots .W_d)$
is a $d$ dimensional standard Brownian motion.
The corresponding Fokker Plank equation describing the evolution of the probability distribution is given by
\begin{equation}\label{fpd}
\partial_t\rho=\nabla \cdot \left[-\rho A+ C\right]
\end{equation}
where
$$
C_i=\frac 12 \sum_{j=1}^d \partial_{x_j}
\left(\rho \left(\mathbb B\mathbb B^T\right)_{i,j}\right)\,, \qquad i=1,\dots ,d\,.
$$
Note that while in the equation \eqref{stoceq} appears the matrix $\mathbb B$, the evolution of the probability distribution
depends just on the symmetric matrix $\mathbb B\mathbb B^T$.
The condition of reversibility (see \cite{G2009, K1981}) is that the vector
\begin{equation}\label{fixA}
F_i:=\sum_{k=1}^d(\mathbb B\mathbb B^T)^{-1}_{i,k}\left[2A_k-\sum_j \partial_{x_j}(\mathbb B\mathbb B^T)_{k,j}\right]\,, \qquad i=1,\dots ,d\,,
\end{equation}
is of gradient type.
In this case, under additional confinements assumptions,  the stationary solution of the Fokker Planck equation is
$$
\bar\rho(x)=\frac{e^{-\psi(x)}}{Z}
$$
where $F=-\nabla \psi$.
We have therefore that all the reversible diffusion processes
can be parameterized in terms of the function $\psi$ and the symmetric and positive definite matrix
$\mathbb B\mathbb B^T$. This is because you can fix arbitrarily these two objects and then $A$ is completely determined by \eqref{fixA}. If we use instead the positive function $\alpha$ related  to $\psi$ by $\psi=\log \alpha$ and the symmetric positive definite matrix
$\mathbb Q(x)=\mathbb B\mathbb B^T(x)\alpha^{-1}(x)/2$ we have that the Fokker Plank equation
\eqref{fpd} is given by
\begin{equation}\label{gemella}
\partial_t\rho=\nabla\cdot \left(\mathbb Q\nabla\left(\alpha \rho\right)\right)
\end{equation}
that is exactly of the type of our hydrodynamic equation \eqref{idro}.
It is important to note however that in \eqref{idro} the matrix $\mathbb Q$ has to be diagonal while instead
this is not the case in \eqref{gemella}. As we will discuss in the next section this is due to the special lattice that we are considering
in Theorem \ref{ilth}. We can obtain non diagonal matrices considering different lattices.

\subsection{Different lattices}
\label{s:diflat}
\par\noindent
Here we show that we obtained just equations with diagonal matrices $\mathbb Q$ since we are considering a squared lattice. We briefly discuss how to handle different situations obtaining non diagonal matrices $\mathbb Q$.
From the
proof of Theorem \eqref{ilth}  we known that the basic computation to identify the limiting equation is to
approximate up to uniformly infinitesimal corrections the term inside square parenthesis in \eqref{ms} that is
\begin{equation}\label{u}
N^2\sum_{y\in C(x)}Q(\{x,y\})\left(f(s,x)-f(s,y)\right)\,.
\end{equation}

The generalized framework that we consider now is a lattice having vertices coinciding again with $\Lambda_N$ but having more edges than the usual square lattice.
This corresponds to allowing more possible jumps to the particles. The graph on which the particles are evolving is obtained as follows. We start with $\mathbb Z^d$
with more edges with respect to the usual ones that are connecting just the minimal distance vertices. The collection of directed edges exiting form any vertex $x\in \mathbb Z^d$
are of the form $(x,x+\tilde v^i)$ where $\tilde v^i$ for $i=1,\dots ,k$ is a collection of vectors such that $x+\tilde v^i\in \mathbb Z^d$. Since
we are always requiring that an un-oriented edge can be crossed on both directions then $k$ has to be necessarily an even number and for any vector $\tilde v^i$
there should be a corresponding label $j$ such that $\tilde v^j=-\tilde v^i$ so that  both $(x,x+\tilde v^i)$ and $(x+\tilde v^i,x)$ are elements of the directed edges $E_N$.
The lattice that we consider is obtained scaling by a factor of $N^{-1}$ this lattice. In particular we call $v^i:=N^{-1}\tilde v^i$.

We have therefore that on each lattice site $x\in \Lambda_N$
there are $k$ different edges incident that correspond to $k$
possible jumps of one particle from x to $x+v^j$, $j=1, \dots ,k$. In the case of the square lattice we had $k=2d$ and each $v^j$ is equal to $\pm e^i$ for some $i$.
Note that we have now
$|C(x)|=k$. More general frameworks are of course possible but for simplicity we restrict to this generalization.

We need to give weights to the vertices and the edges of the lattice suitably discretizing smooth objects. The weights on the vertices are associated as before computing a smooth function $\alpha$ on the corresponding point. For the edges we need to generalize the construction done before.

We consider a smooth metrics  $\mathcal Q(x)$ that is a symmetric and positive definite $d \times d$ matrix depending in a regular way ($C^2$ for example) on the continuous variable $x\in \Lambda$. We associate the weight to an edge of the form
$\{x,x+v^i\}$ as
\begin{equation}\label{pesiqq}
Q(\{x,x+v^i\}):=\tilde v^i\cdot \mathcal Q(x+v^i/2)\tilde v^i=:Q^i(x+v^i/2)\,.
\end{equation}
The appearance of the $\tilde v^i$ vectors above is due to the fact that we have $|v^i|\sim 1/N$ (since the vectors without tilde are comparable with the mesh of the lattice)
and we want that the weights to be associated to the edges are not infinitesimal in $N$ but are of order one. The last equality in \eqref{pesiqq}
is just the definition of a shorthand for the weights.
With a suitable Taylor expansion we get that \eqref{u} coincides up to uniformly infinitesimal terms with
\begin{equation}\label{denr}
N^2\sum_{i=1}^k\left(Q^{i}(x)+\nabla Q^{i}(x)\cdot\frac{v^i}{2}\right)\left(\nabla f(x)\cdot v^i+\frac{1}{2}v^i\cdot H(x)v^i\right)\,,
\end{equation}
where $H(x)$ is the Hessian matrix at $x$ of the function $f$ having elements $\left(H(x)\right)_{l,m}=\partial_{x_l}\partial_{x_m}f(x)$.

Recall that $k$ is an even number ad if $v^i$ is the vector associated to a possible jump then also $-v^i$ is a vector associated to a possible jump. Due to this, we
have that the leading term in the product in \eqref{denr}
that is
\begin{equation}\label{red}
N^2\sum_{i=1}^kQ^{i}(x)\nabla f(x)\cdot v^i
\end{equation}
is identically zero. This is because we can pair the edges exiting from $x$ in such a way that if the label $i$ is paired to the label $j$ then $v^i=-v^j$ and consequently $Q^{i}(x)=Q^{j}(x)$. Of the remaining three terms obtained when we develop the product in \eqref{denr} we have that one is infinitesimal.
The two relevant ones that survive are
$$
\frac 12\sum_{i=1}^kQ^{i}(x)Nv^i\cdot H(x)Nv^i+\frac 12\sum_{i=1}^k\left(\nabla Q^{i}(x)\cdot N v^i\right)\left(\nabla f(x)\cdot Nv^i\right)\,.
$$
The above expression coincides up to uniform infinitesimal terms with
$$
\nabla\cdot \left(\mathbb Q(x)\nabla f(x)\right)
$$
where the matrix $\mathbb Q$ is defined as
\begin{equation}\label{Qnonodiag}
\mathbb Q_{l,m}(x)=\frac 12\sum_{i=1}^k Q^{i}(x)\tilde v^i_l\tilde v^i_m\,.
\end{equation}
With the same arguments of the proof of Theorem \ref{ilth}, but using this expansion, we can prove that the limiting equation
is again of the form \eqref{idro} but the matrix $\mathbb Q$ is given by \eqref{Qnonodiag} that in general is non--diagonal.

\subsection{Uphill currents}
\label{s:prb}
\par\noindent
A current is said to move ``uphill'' when particles migrate \textit{up the gradient}, namely towards regions of higher concentration, thus violating the basic tenets of Fick's law of diffusion. The onset of such uphill currents can be traced back to the action of an external field, to the presence of mutual interactions in a multi-component system or, for single-component systems, to  a phase transition, and was recently investigated in a variety of lattice gas models, cf. Refs \cite{CDMP16,CDMP17,CDMP17bis,CC2017,CGGV18}.

We look, here, at the case where two inhomogeneous diffusion processes take place in two intervals of length $\Len>0$, for two concentration functions $\uf$, $\uk$, being connected by conditions of equality of concentration and of flux at the two endpoints. The latter is meant in the sense that the outflux of $\uf$ equals the influx of $\uk$. However, $\uf$ solves Fick's equation, while $\uk$ solves a Fokker-Planck type equation. The diffusivities are assumed to be piecewise constant.

We consider the stationary case, see also Refs. for a more general discussion about the observation of uphill currents.

Thus the problem is, in a distributional formulation,
\begin{alignat}2
  \label{eq:pde_f}
  -(\Df \uf_{x})_{x}
  &=
  0
  \,,
  &\qquad&
  0<x<\Len
  \,,
  \\
  \label{eq:pde_k}
  -(\Dk \uk)_{xx}
  &=
  0
  \,,
  &\qquad&
  0<x<\Len
  \,,
  \\
  \label{eq:conc_0}
  \uf(0)
  &=
  \uk(0)
  \,,
  &\qquad&
  \\
  \label{eq:conc_L}
  \uf(\Len)
  &=
  \uk(\Len)
  \,,
  &\qquad&
  \\
  \label{eq:flux_0}
  \Df \uf_{x}(0)
  &=
  -
  \Dk \uk_{x}(0)
  \,,
  &\qquad&
  \\
  \label{eq:flux_L}
  \Df \uf_{x}(\Len)
  &=
  -
  \Dk \uk_{x}(\Len)
  \,.
  &\qquad&
\end{alignat}
Here
\begin{equation}
  \label{eq:Df}
  \Df(x)
  =
  \Df_{1}
  \chi_{(0,\Lf)}(x)
  +
  \Df_{2}
  \chi_{(\Lf,\Len)}(x)
  \,,
\end{equation}
and
\begin{equation}
  \label{eq:Dk}
  \Dk(x)
  =
  \Dk_{1}
  \chi_{(0,\Lk)}(x)
  +
  \Dk_{2}
  \chi_{(\Lk,\Len)}(x)
  \,,
\end{equation}
for given positive constants $\Df_{i}$, $\Dk_{i}$, and for $\Lf$, $\Lk\in(0,\Len)$.

We assume here $\Dk_{1}\not=\Dk_{2}$; see also Remark~\ref{r:trivial}.

We refer to the following weak formulation of
this problem: find $\uf\in H^{1}(0,L)$, $\uk\in L^{\infty}(0,L)$ such
that $\Dk \uk\in H^{1}(0,L)$ and
\begin{equation}
  \label{eq:pde_weak}
  \int_{0}^{L}
  \{
  \Df
  \uf_{x}
  \zeta_{x}
  +
  (\Dk\uk)_{x}
  \eta_{x}
  \}
  d x
  =
  0
  \,,
\end{equation}
for all $\zeta$, $\eta\in C^{1}([0,L])$ such that $\zeta(0)=\eta(0)$
and $\zeta(L)=\eta(L)$. Here $H^{1}(0,L)$ is the standard space of
square integrable functions with square integrable Sobolev derivative,
which is known to be embedded in $C([0,L])$. Then, also using our
assumptions on $\Dk$, we impose \eqref{eq:conc_0} and
\eqref{eq:conc_L} in a classical pointwise sense.

It follows from straightforward reasoning and from \eqref{eq:pde_weak}
that $K\uf_{x}$ and $(\Dk\uk)_{x}$ are constant in $(0,L)$. Thus
invoking the definitions of $\Df$ and $\Dk$, we recover in the classical sense
\begin{alignat}2
  \label{eq:pde2_f}
  -
  \uf_{xx}
  &=
  0
  \,,
  &\qquad&
  \text{in $(0,\Lf)\cup(\Lf,\Len)$,}
  \\
  \label{eq:cint_f}
  \uf(\Lf-)
  &=
  \uf(\Lf+)
  \,,
  &\qquad&
  \\
  \label{eq:fint_f}
  \Df_{1} \uf_{x}(\Lf-)
  &=
  \Df_{2} \uf_{x}(\Lf+)
  \,,
  &\qquad&
\end{alignat}
and
\begin{alignat}2
  \label{eq:pde2_k}
  -
  \uk_{xx}
  &=
  0
  \,,
  &\qquad&
  \text{in $(0,\Lk)\cup(\Lk,\Len)$,}
  \\
  \label{eq:cint_k}
  \Dk_{1}
  \uk(\Lk-)
  &=
  \Dk_{2}
  \uk(\Lk+)
  \,,
  &\qquad&
  \\
  \label{eq:fint_k}
  \Dk_{1} \uk_{x}(\Lk-)
  &=
  \Dk_{2} \uk_{x}(\Lk+)
  \,.
  &\qquad&
\end{alignat}
Note that more generally one should write e.g., \eqref{eq:fint_k} as
\begin{equation*}
  (\Dk_{1} \uk)_{x}(\Lk-)
  =
  (\Dk_{2} \uk)_{x}(\Lk+)
  \,.
\end{equation*}
but this is not relevant under our assumption of piecewise constant
$\Dk$. A similar remark applies to \eqref{eq:flux_0},
\eqref{eq:flux_L}, which indeed are valid in a pointwise sense.

Clearly problem \eqref{eq:pde_f}--\eqref{eq:flux_L} is invariant for multiplication by a constant, and always has the null solution. Therefore for the sake of precision we'll impose also the following normalization condition
\begin{equation}
  \label{eq:norm}
  \uf(0)
  =
  1
  \,.
\end{equation}
The formulation \eqref{eq:pde2_f}--\eqref{eq:fint_f} yields immediately
\begin{equation}
  \label{eq:uf_pcw}
  \uf(x)
  =
  \left\{
  \begin{alignedat}2
    &\uf_{x}(\Lf-)
    (x-\Lf)
    +
    \uf(\Lf-)
    \,,
    &\qquad&
    0<x<\Lf
    \,,
    \\
    &
    \frac{\Df_{1}}{\Df_{2}}
    \uf_{x}(\Lf-)
    (x-\Lf)
    +
    \uf(\Lf-)
    \,,
    &\qquad&
    \Lf<x<\Len
    \,.
  \end{alignedat}
  \right.
\end{equation}
Instead the formulation \eqref{eq:pde2_k}--\eqref{eq:fint_k} implies
\begin{equation}
  \label{eq:uk_pcw}
  \uk(x)
  =
  \left\{
  \begin{alignedat}2
    &\uk_{x}(\Lk-)
    (x-\Lk)
    +
    \uk(\Lk-)
    \,,
    &\qquad&
    0<x<\Lk
    \,,
    \\
    &
    \frac{\Dk_{1}}{\Dk_{2}}
    \uk_{x}(\Lk-)
    (x-\Lk)
    +
    \frac{\Dk_{1}}{\Dk_{2}}
    \uk(\Lk-)
    \,,
    &\qquad&
    \Lk<x<\Len
    \,.
  \end{alignedat}
  \right.
\end{equation}
The normalization condition and \eqref{eq:conc_0} lead to
\begin{align}
  \label{eq:sys_uf0}
  -
  \uf_{x}(\Lf-)
  \Lf
  +
  \uf(\Lf-)
  &=
    1
    \,,
  \\
  \label{eq:sys_uk0}
  -
  \uk_{x}(\Lk-)
  \Lk
  +
  \uk(\Lk-)
  &=
    1
    \,,
\end{align}
while \eqref{eq:conc_L} gives
\begin{equation}
  \label{eq:sys_ufkL}
  \frac{\Df_{1}}{\Df_{2}}
  \uf_{x}(\Lf-)
  (\Len-\Lf)
  +
  \uf(\Lf-)
  =
  \frac{\Dk_{1}}{\Dk_{2}}
  \uk_{x}(\Lk-)
  (\Len-\Lk)
  +
  \frac{\Dk_{1}}{\Dk_{2}}
  \uk(\Lk-)
  \,.
\end{equation}
Finally both \eqref{eq:flux_0} and \eqref{eq:flux_L} are equivalent to
\begin{equation}
  \label{eq:sys_flux}
  \Df_{1}
  \uf_{x}(\Lf-)
  =
  -
  \Dk_{1}
  \uk_{x}(\Lk-)
  \,.
\end{equation}
Thus we have a linear system \eqref{eq:sys_uf0}--\eqref{eq:sys_flux} of 4 equations in the 4 unknowns
$\uf(\Lf-)$, $\uf_{x}(\Lf-)$, $\uk(\Lk-)$, $\uk_{x}(\Lk-)$.

Its solution is
\begin{align*}
  \uf(\Lf-)
  &=
    \frac{
    (\Dk_{1}\Df_{2}-\Dk_{2}\Df_{1})
    \Lf
    +
    \Df_{1}
    (\Dk_{2}+\Df_{2})
    \Len
    }{
    \Dk_{2}
    (\Df_{2}-\Df_{1})
    \Lf
    +
    \Df_{1}
    (\Dk_{2}+\Df_{2})
    \Len
    }
    =
    1
    +
    \frac{
    \Df_{2}
    (\Dk_{1}-\Dk_{2})
    \Lf
    }{
    \Dk_{2}
    (\Df_{2}-\Df_{1})
    \Lf
    +
    \Df_{1}
    (\Dk_{2}+\Df_{2})
    \Len
    }
    \,,
  \\
  \uf_{x}(\Lf-)
  &=
    \frac{
    \Df_{2}
    (\Dk_{1}-\Dk_{2})
    }{
    \Dk_{2}
    (\Df_{2}-\Df_{1})
    \Lf
    +
    \Df_{1}
    (\Dk_{2}+\Df_{2})
    \Len
    }
    \,,
  \\
  \uk(\Lk-)
  &=
    1
    +
    \frac{1}{\Dk_{1}}
    \,
    \frac{
    \Df_{1}
    \Df_{2}
    (\Dk_{2}-\Dk_{1})
    \Lk
    }{
    \Dk_{2}
    (\Df_{2}-\Df_{1})
    \Lf
    +
    \Df_{1}
    (\Dk_{2}+\Df_{2})
    \Len
    }
    \,,
  \\
  \uk_{x}(\Lk-)
  &=
    \frac{1}{\Dk_{1}}
    \,
    \frac{
    \Df_{1}
    \Df_{2}
    (\Dk_{2}-\Dk_{1})
    }{
    \Dk_{2}
    (\Df_{2}-\Df_{1})
    \Lf
    +
    \Df_{1}
    (\Dk_{2}+\Df_{2})
    \Len
    }
    \,,
\end{align*}
provided
\begin{equation*}
  \Dk_{2}
  (\Df_{2}-\Df_{1})
  \Lf
  +
  \Df_{1}
  (\Dk_{2}+\Df_{2})
  \Len
  \not=
  0
  \,.
\end{equation*}
But
\begin{equation*}
  \Dk_{2}
  (\Df_{2}-\Df_{1})
  \Lf
  +
  \Df_{1}
  (\Dk_{2}+\Df_{2})
  \Len
  =
  \Dk_{2}
  \Df_{2}
  \Lf
  +
  \Df_{1}
  \Df_{2}
  \Len
  +
  \Df_{1}
  \Dk_{2}
  (\Len-\Lf)
  >0
  \,,
\end{equation*}
since $\Len>\Lf$.

We remark that each one of $\uf_{x}(x)$, $x\not=\Lf$, and $\uk_{x}(x)$, $x\not=\Lk$, has constant sign; the two signs always differ.
This remark does not imply that $\uk$ is monotonic, in view of its discontinuous character.

We may also compute
\begin{equation*}
  \uf(L)
  =
  \uk(L)
  =
  1
  +
  (\Dk_{1}-\Dk_{2})
  \frac{
    (\Df_{2}-\Df_{1})
    \Lf
    +
    \Df_{1}
    \Len
  }{
    \Dk_{2}
    (\Df_{2}-\Df_{1})
    \Lf
    +
    \Df_{1}
    (\Dk_{2}+\Df_{2})
    \Len
  }
  >0
  \,,
\end{equation*}
where the last inequality follows from elementary reasoning.

\begin{figure}
\begin{picture}(80,140)(-10,0)
\includegraphics[width=0.4\textwidth]{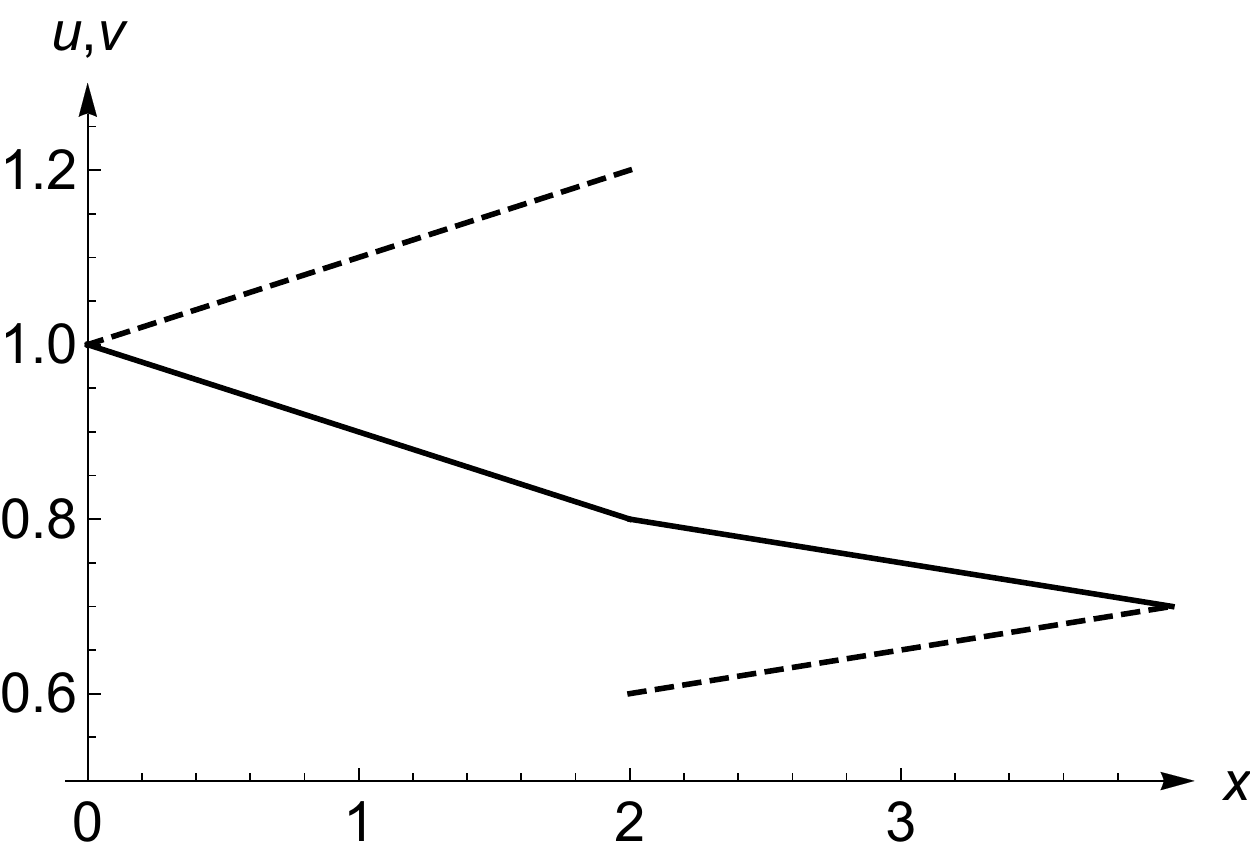}
\hskip 1. cm
\includegraphics[width=0.4\textwidth]{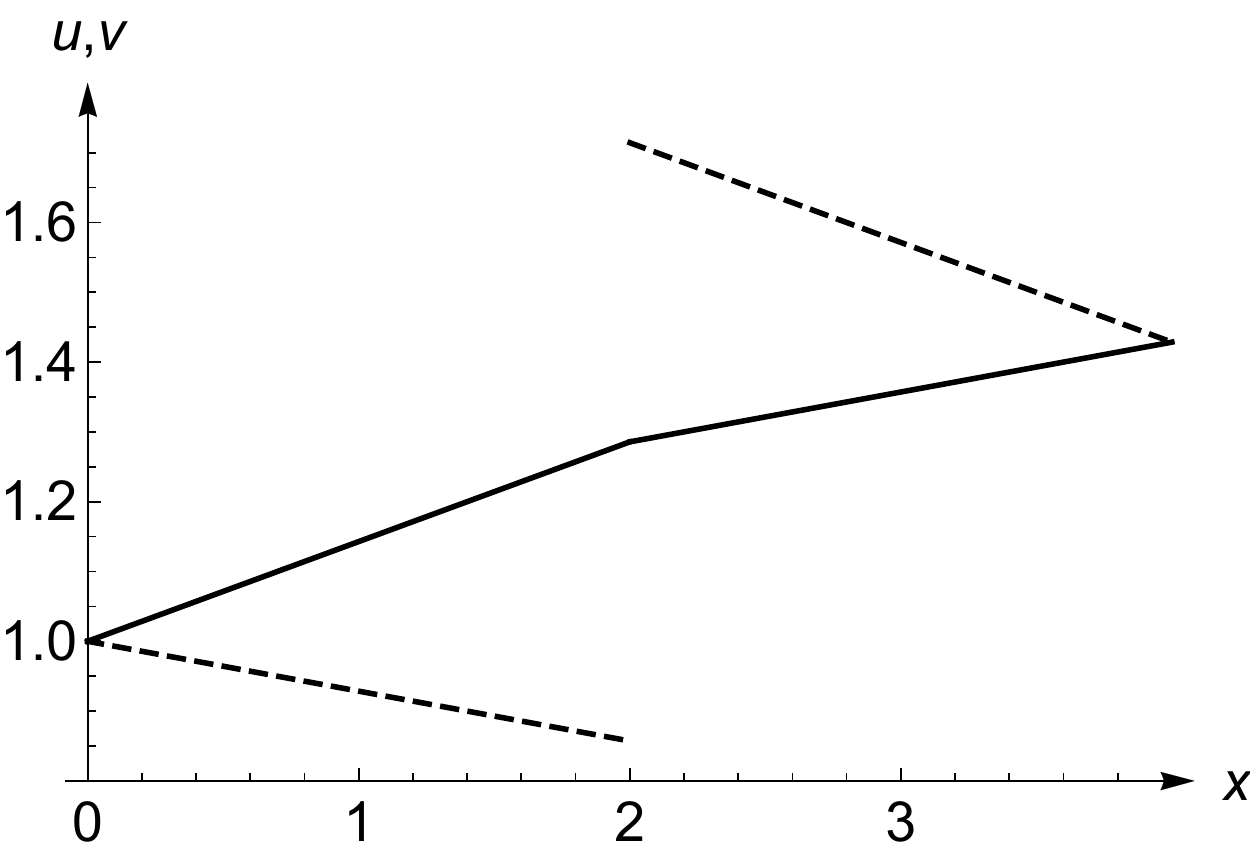}
\end{picture}
\caption{Functions $v$ (continuous line) and $u$ (dashed line)
for
$\Len=4$, $\Lf=\Lk=2$, $\Df_{1}=1$, and $\Df_{2}=2$
with
$\Dk_{1}=1$ and $\Dk_{2}=2$ on the left and
$\Dk_{1}=2$ and  $\Dk_{2}=1$ on the right.}
  \label{f:two}
\end{figure}

\begin{remark}
  \label{r:trivial}
  If $\Dk_{1}=\Dk_{2}$ one can see easily that the solution is flat, that is $\uf(x)=\uk(x)=1$ for all $x\in(0,\Len)$.
  This is a special case of next Remark~\ref{r:fick2fick}.

  Instead the relative values of $\Df_{1}$, $\Df_{2}$ do not seem to play any special role.
\end{remark}

\begin{remark}
  \label{r:fick2fick}
  If one assumes for $\uk$ a Ficksian equation similar to the one solved by $\uf$, it follows immediately that $\uf(x)=\uk(x)=1$ for all $x\in(0,\Len)$: indeed since both $\uf$ and $\uk$ are continuous and piecewise linear, and then monotonic, they share their minimum and maximum values, at the endpoints. But there their fluxes are opposite in sign, and must therefore actually vanish, yielding the claim.
\end{remark}

\begin{remark}
  \label{r:partition}
  If we replace the conditions \eqref{eq:conc_0}, \eqref{eq:conc_L} with the partition type balances
  \begin{align}
    \label{eq:conc_0_part}
    \uf(0)
    &=
    \Dk\uk(0)
    \,,
    \\
    \label{eq:conc_L_part}
    \uf(\Len)
    &=
    \Dk\uk(\Len)
    \,,
  \end{align}
  it can be immediately seen that setting $\tilde\uk=\Dk\uk$ we obtain for $\uf$, $\tilde\uk$ a problem with two equations of Fick type; more exactly we are in the case of Remark~\ref{r:fick2fick} with the diffusivity in the equation for $\tilde\uk$ being identically $1$. Then we have
  \begin{equation}
    \label{eq:partition}
    \uf(x)
    =
    1
    \,,
    \qquad
    \Dk(x)
    \uk(x)
    =
    1
    \,,
    \qquad
    x\in (0,L)
    \,.
  \end{equation}
  On the other hand, conditions \eqref{eq:conc_0_part} are comparable to \eqref{eq:cint_k}; that is they are the conditions we would expect if the whole system was subject to the equation
  \begin{equation*}
    -
    (K(x)(D(x)U(x))')'
    =
    0
    \,,
  \end{equation*}
  with the suitable choices of $K$, $D$.
\end{remark}



\begin{thebibliography}{100}


\bibitem{ABC2014}
D.\ Andreucci, D.\ Bellaveglia, E.N.M.\ Cirillo,
\textit{A model for enhanced and selective transport through biological
	membranes with alternating pores.}
Mathematical Biosciences \textbf{257}, 42--49 (2014).

\bibitem{MFTReview}
L. Bertini, A. De Sole,  D. Gabrielli, G. Jona-Lasinio, C. Landim {\it Macroscopic fluctuation theory},
Rev. Mod. Phys. {\bf 87}, 593 (2015).


\bibitem{BG2005}
N.H.\ Bian, O.E.\ Garcia,
\textit{Structures, profile consistency, and transport scaling in
	electrostatic convection.}
Physics of Plasmas \textbf{12}, 042307 (2005).






\bibitem{Bil} Billingsley \textit{Convergence of probability measures, Second edition}
Wiley Series in Probability and Statistics (1999).

\bibitem{CC2017}
E.N.M.\ Cirillo,\ M. Colangeli,
\textit{Stationary uphill currents in locally perturbed Zero Range Processes.}
Physical Review E \textbf{96}, 052137 (2017).

\bibitem{CDMR}
E.N.M.\ Cirillo, I.\ De Bonis, A.\ Muntean, O.\ Richardson,
\textit{Driven particle flux through a membrane: Two--scale asymptotics of a
	diffusion equation with polynomial drift.}
Preprint 2018, arXiv:1804.08392.

\bibitem{CKMS2016}
E.N.M.\ Cirillo, O.\ Krehel,
A.\ Muntean, R.\ van Santen,
\textit{A lattice model of reduced jamming by barrier.}
Physical Review E \textbf{94}, 042115 (2016).

\bibitem{CKMSS2016}
E.N.M.\ Cirillo, O.\ Krehel,
A.\ Muntean, R.\ van Santen, A.\ Sengar,
\textit{Residence time estimates for asymmetric simple exclusion dynamics on
strips.}
Physica A \textbf{442}, 436--457 (2016).






\bibitem{CDMP16} M.~Colangeli, A.~De~Masi, and E.~Presutti,
\textit{Latent heat and the Fourier law.}
{\em Physics Letters A} \textbf{380}, 1710--1713 (2016);


\bibitem{CDMP17} M.~Colangeli, A.~De~Masi, and E.~Presutti,
\textit{Particle models with self-sustained current.}
{\em J. Stat. Phys.} \textbf{167}, 1081--1111 (2017).

\bibitem{CDMP17bis} M.~Colangeli, A.~De Masi, and E.~Presutti,
\textit{Microscopic models for uphill diffusion.}
{\em J. Phys. A: Math. Theor.} \textbf{50}, 435002 (2017).


\bibitem{CGGV18} M.~Colangeli, C.~Giardin\`{a}, C.~Giberti and C.~Vernia,
\textit{Non-equilibrium 2D Ising model with stationary uphill diffusion.}
{\em Phys. Rev. E} \textbf{97}, 030103(R) (2018).

\bibitem{CCM1997}
R.\ Collins, S.R.\ Carson, J.A.D.\ Matthew,
\textit{Diffusion equation for one--dimensional unbiased hopping.}
American Journal od Physics \textbf{65}, 230 (1997).

\bibitem{DmP1991}
A.\ De Masi, E.\ Presutti,
\textit{Mathematical Methods for Hydrodynamic Limits.}
Springer--Verlag, Berlin Heidelberg (1991).


\bibitem{G2009}
C.\ Gardiner,
\textit{Stochastic methods.}
Springer--Verlag, Berlin Heidelberg, 2009.

\bibitem{GDISP2006}
K.\ Ghosh, K.A.\ Dill, M.M.\ Inamdar, E.\ Seitaridou, R.\ Phillips,
\textit{Teaching the principles of statistical dynamics.}
American Journal of Physics \textbf{74}, 123 (2006).

\bibitem{K1981}
N.G.\ van Kampen,
\textit{Stochastic processes in physics and chemistry.}
North--Holland, 1981.

\bibitem{KL1999}
C.\ Kipnis, C.\ Landim,
\textit{Scaling Limits of Interacting Particle Systems.}
Springer--Verlag Berlin Heidelberg, 1999.

\bibitem{LBLO2001}
P.\ Lan\c{c}on, G.\ Batrouni, L.\ Lobry, N.\ Ostrowsky,
\textit{Drift without flux: Brownian walker with a space--dependent
diffusion coefficient.}
Europhysics Letters \textbf{54}, 58--34 (2001).

\bibitem{LSU}
O.A.\ Ladyzhenskaja, V.A.\ Solonnikov, N.N.\ Ural'ceva,
\textit{Linear and Quasilinear Equations of Parabolic Type},
American Mathematical Society,
Providence, RI (1968).

\bibitem{L1984}
P.T.\ Landsberg,
\textit{$D\textrm{grad}\,v$ or $\textrm{grad}(Dv)$?}
Journal of Applied Physics \textbf{56}, 1119 (1984).

\bibitem{MBCS2005}
B.Ph.\ van Milligen, P.D.\ Bons, B.A.\ Carreras, R.\ S\'anchez,
\textit{On the applicability of Fick's law to diffusion
	in inhomogeneous systems.}
European Journal od Physics \textbf{26}, 913--925 (2005).


\bibitem{asLD}  F. Rassoul--Agha , T. Sepp\"{a}l\"{a}inen  \emph{A course on large deviations with an introduction to Gibbs measures} Graduate Studies in Mathematics, {\bf 162} American Mathematical Society, Providence, RI, (2015)

\bibitem{S2008}
F.\ Sattin,
\textit{Fick's law and Fokker--Planck equation in inhomogeneous
environments.}
Physics Letters A \textbf{372}, 3921--3945 (2008).


\bibitem{S1993}
M.J.\ Schnitzer,
\textit{Theory of continuum random walks and application to
	chemotaxis.}
Physical Review E \textbf{48}, 2553--2568 (1993).


\bibitem{SBBP1990}
M.J.\ Schnitzer, S.M.\ Block, H.C.\ Berg, E.M.\ Purcell,
\textit{Strategies for chemotaxis.}
Symp.\ Soc.\ Gen.\ Microbiology \textbf{46}, 15 (1990).

\bibitem{SD2005}
Y.H.\ Sniekers, C.C.\ van Donkelaar,
\textit{Determining Diffusion Coefficients in Inhomegeneous
Tissue Using Fluorescence recovery after Photobleaching.}
Biophysical Journal \textbf{89}, 1302--1307 (2005).

\bibitem{Spohn}
H. Spohn, \emph{Large Scale Dynamics of Interacting Particles}
Springer-Verlag, New York (1991).

\end{thebibliography}

\end{document}